\documentclass[12pt, draftclsnofoot,onecolumn]{IEEEtran}

\usepackage{cite, graphicx}
\usepackage{hyperref}
\usepackage[cmex10]{amsmath}
\usepackage{amsthm}
\usepackage{amsmath}
\allowdisplaybreaks[3]
\usepackage{amssymb}
\ifCLASSOPTIONcompsoc
\usepackage[caption=false,font=normalsize,labelfont=sf,textfont=sf]{subfig}
\else
\usepackage[caption=false,font=footnotesize]{subfig}
\fi
\usepackage{algorithmic}
\usepackage{algorithm}
\usepackage{mathrsfs}
\usepackage{stfloats}
\usepackage{float}
\usepackage{setspace}
\usepackage{cases}
\usepackage{xcolor}

\theoremstyle{plain}
\newtheorem{thm}{Theorem}

\newtheorem{prop}[thm]{Proposition}

\theoremstyle{definition}

\theoremstyle{remark}

\begin{document}

\title{Flexible Functional Split and Power Control for Energy Harvesting Cloud Radio Access Networks}

\author{Liumeng Wang, Sheng Zhou
\thanks{Part of this work has been presented in IEEE ICC 2018 Workshop \cite{mine}. This work is sponsored in part by the National Key R\&D Program of China 2018YFB1800804, Nature Science Foundation of China (No. 61571265, No. 61871254, No. 91638204, No. 61861136003, No. 61621091), and Hitachi Ltd. }
\thanks{Liumeng Wang and Sheng Zhou are with Beijing National Research Center for Information Science and Technology,
        Department of Electronic Engineering, Tsinghua University, Beijing 100084, China (Email: wlm14@mails.tsinghua.edu.cn, sheng.zhou@tsinghua.edu.cn).}}

\maketitle

\begin{abstract}
Functional split is a promising technique to flexibly balance the processing cost at remote ends and the fronthaul rate in cloud radio access networks (C-RAN). By harvesting renewable energy, remote radio units (RRUs) can save grid power and be flexibly deployed. However, the randomness of energy arrival poses a major design challenge. To maximize the throughput under the average fronthaul rate constraint in C-RAN with renewable powered RRUs, we first study the offline problem of selecting the optimal functional split modes and the corresponding durations, jointly with the transmission power. We find that between successive energy arrivals, at most two functional split modes should be selected. Then the optimal online problem is formulated as an Markov decision process (MDP). To deal with the curse of dimensionality of solving MDP, we further analyze the special case with one instance of energy arrival and two candidate functional split modes as inspired by the offline solution, and then a heuristic online policy is proposed. Numerical results show that with flexible functional split, the throughput can be significantly improved compared with fixed functional split. Also, the proposed heuristic online policy has similar performance with the optimal online one, as validated by simulations.
\end{abstract}

\begin{IEEEkeywords}
Functional split, energy harvesting, cloud radio access network (C-RAN), fronthaul, Markov decision process (MDP).
\end{IEEEkeywords}

\section{Introduction}
\label{sec:intro}
Cloud radio access network (C-RAN) \cite{CRANwhitepaper}, which centralizes the baseband functions at the baseband units (BBUs), can efficiently reduce the complexity of the remote radio units (RRUs), and thus the operation and deployment costs. Centralized baseband processing also enables efficient cooperative signal processing to increase the network capacity. In C-RAN, the fronthaul network transports the baseband signals between the BBUs and the RRUs. However, for fully centralized C-RAN, i.e., all baseband functions are centralized at the BBUs, the fronthaul rate requirement is high, which poses a major design challenge on C-RAN. For example, in a single 20MHz LTE antenna-carrier system, 1Gbps fronthaul rate is required with the standard CPRI interface \cite{CPRI}. To support massive MIMO and other emerging technologies, the required fronthaul rate will be too high to bear.

\begin{figure}
	\centering
	\includegraphics[width=0.65\textwidth]{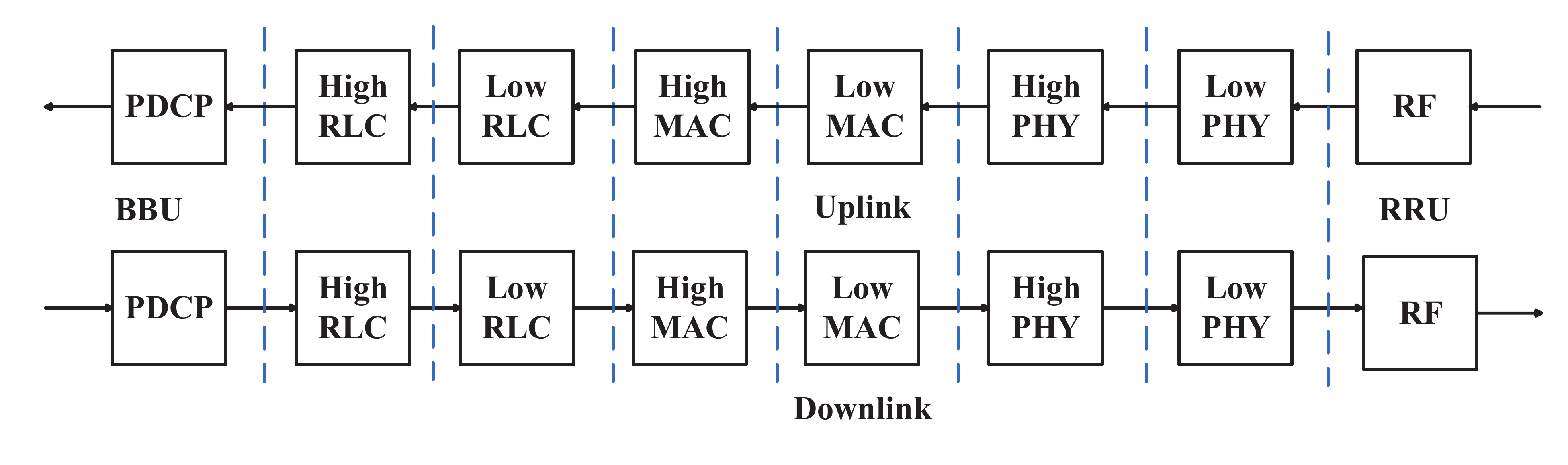}
	\caption{Illustration of baseband functions with multiple candidate split modes.}
	\label{fig:functionalsplit}
\end{figure}

Different from fully centralized C-RAN, by placing some baseband and network functions at RRUs, functional split is a promising technique to reduce the fronthaul rate requirement \cite{SplittingBS, impact}. There are multiple candidate functional split modes corresponding to different split points on the chain of baseband functions, as illustrated in Fig. \ref{fig:functionalsplit}. For each mode, the functions placed at the right side of the corresponding vertical dashed line are placed at the RRU, while the others are centralized at the BBU. The fronthaul rate requirement and processing complexity requirement at the RRUs vary, under different functional split modes. In general, with more baseband functions at the RRUs, the required fronthaul rate is lower, but the processing complexity is higher \cite{LTEmodel, centralize}, which also means more energy consumption at the RRUs. With certain functional split modes, for example, split between the physical layer and the MAC layer, the required fronthaul rate depends on the traffic load, and thus exploiting the fronthaul statistical multiplexing gain can further reduce the fronthaul rate requirement \cite{redesign, multiplexing}. With the development of software defined network (SDN) and  network function virtualization (NFV), baseband functions can be virtualized and implemented on the general purpose computation platforms \cite{SDN, SDRandFS}. As a result, the functions placed at the RRUs and the BBUs can be reconfigured according to the network state \cite{Flex5G, Flex5GSurvey}.

By harvesting renewable energy from the environment, the RRUs are able to consume less or no energy from the power grid \cite{EHNodes, grid, greendelivery}. Another benefit is that the RRUs can be flexibly deployed at the places where the grid does not reach.
However, reliable communication is challenging due to the randomness of renewable energy arrivals and the limitation of batteries, and thus the operation of RRUs should be well managed \cite{EHSurvey}. In terms of power control, different from conventional ``water-filling", the throughput-optimal ``directional water-filling" power control policy is found in a fading energy harvesting channel\cite{waterfilling}, where the ``water", i.e., the energy, can only flow from the past to the future. If the processing energy consumption is considered, the throughput-optimal transmission policy should become bursts, a ``glue pouring" power control policy is proved to be optimal when there is only one energy arrival and no transmission deadline \cite{glue}. The burst transmission is due to the fact that more processing energy is consumed with longer transmission time.
For energy harvesting system with processing cost and multiple energy arrivals, a ``directional backward glue-pouring" algorithm is proposed in \cite{procost}.

{There are some recent works on the flexible functional split mode selection in energy harvesting C-RAN systems. The grid power consumption and system outage rate are jointly studied by optimizing the offline placement of baseband functions, where the small base station is powered by renewable energy and the macro base station is powered by the grid \cite{split1}. Reinforcement learning based  online placement of functional split options is studied in \cite{split2} for efficient utilization of the harvested energy, where the small cell is powered by renewable power with flexible functional split modes.
To improve energy efficiency and throughput, RRU active/sleep mode and functional split mode selection in the energy harvesting C-RAN are determined according to the renewable energy levels and the number of users in the covering area of the RRU \cite{sleepfs}.
However, to the best of our knowledge, the joint optimization of power control and flexible functional split mode selection has not been considered yet.
}

If the functional split mode is fixed in the energy harvesting communication system, the processing power is a constant, and thus ``directional backward glue-pouring" algorithm \cite{procost} can be used to find the optimal power control policy. However, it is expensive and sometimes difficult to deploy fibers between the RRUs and the BBUs, and thus wireless fronthaul may be used as a low cost solution \cite{wirelessBH}. Especially for RRUs powered by energy harvesting, they in general have no wired connection neither for power supply nor for fronthaul. In this case, flexible functional split is necessary, due to not only the fronthauling overhead brought by the wireless fronthaul, but also the unstable renewable energy supply. To this end, there are more than one candidate functional split modes, with different processing costs, and thus existing schemes like ``directional backward glue-pouring" no longer apply. Functional split can tradeoff between the baseband processing complexity of RRUs and the fronthaul data rate requirement. In general, with more baseband functions at the RRU, the baseband processing complexity is higher, but the required fronthaul data rate is lower. Conversely, with more baseband functions at the RRU, the baseband processing power is lower, but the required fronthaul data rate is higher. This calls for new mechanisms that can determine the optimal functional split with the joint consideration of fronthaul properties and renewable energy arrivals.

In this paper, we study the selection of the functional split modes and power control policy for an energy harvesting RRU in C-RAN. We first consider the offline case, where the energy arrivals and the channel fading are non-causally known in advance. The functional split is jointly determined with the corresponding user data transmission duration and transmission power,
and the objective is to maximize the throughput, while satisfying the energy and the average fronthaul rate constraints. For the optimal offline policy, we find that in each interval between successive energy arrivals, at most two modes are selected, the transmission power of the modes are the same for each channel fading block. We further analyze the scenarios with only one instance of energy arrival and two alternative functional split modes, and get the closed-from expression of the transmission power and transmission duration for each split mode, given the average fronthaul rate constraint. Based on the analysis, we propose a heuristic online policy, where the future energy arrivals and the channel fading are unknown in advance. Numerical results show that the heuristic online policy has similar performance with the optimal online policy developed by solving the Markov decision process (MDP) formulation.

The main contributions of this paper are summarized as follows.
\begin{itemize}
\item We jointly optimize the functional split mode selection and power control for an RRU powered with renewable energy, to maximize the throughput under the average fronthaul rate constraint and random energy arrival.
\item For the offline problem where the energy arrivals and the channel fading are non-causally known, the throughput maximization problem is formulated and analyzed. We find the structure of the optimal solution that at most two functional split modes are selected between two successive energy arrivals. The online problem where the channel fading are causally known, is solved by its corresponding MDP formulation through value iteration.
\item To deal with the curse of dimensionality in solving an MDP, the closed-form expression of the transmission power and transmission duration in the special case with one energy arrival is derived, based on which a low-complexity heuristic online policy is proposed, and is shown to have near-optimal performance via extensive simulations.
\end{itemize}

The paper is organized as follows. The system model is described in Section \ref{sec:sysmodel}. The offline optimization problem is formulated and analyzed in Section \ref{sec:maximize}, and the online problem is introduced and solved by an MDP formulation in Section \ref{sec:mdp}. The expression of optimal power control policy with one energy arrival, two functional split modes is derived in Section \ref{sec:single}. A heuristic online policy is proposed in Section \ref{sec:online}. The numerical results are presented in Section \ref{sec:num}. The paper is concluded in Section \ref{sec:conclusion}.

\section{System Model} \label{sec:sysmodel}
{
Consider a two-tier network, where a macro base station (MBS) covers a large area, while an RRU has small coverage areas within the coverage area of
the MBS. The MBS has stable power supply, while the RRU is powered by renewable energy.
The RRU transmits as much data as possible to the users with the harvested energy, while the remaining data is transmitted via the MBS.
We thus aim to maximize the throughput of the RRU to reduce the traffic load of the MBS.
We consider the downlink transmission from a particular RRU to its users, as described in Fig. \ref{fig:mdpmodel}(a).
Assume that the BBU has sufficient data to transmit to users.
}

The system is slotted with normalized slot length.
Assume that the wireless channel of the users is block fading, where the channel gain varies every block but remains constant within one block. Each block has $L$ slots. For each slot, the RRU serves the user with the best channel state, i.e., the user with the largest channel gain.

We assume that the energy arrives over a larger time granularity than that of the wireless channel fading  \cite{Huang14, Gong18}. The energy arrival rate stays constant in $N$ blocks, which is denoted by an epoch. We assume that the energy only arrives at the beginning of each epoch. The approximation is adopted to analyze the effect of different time scales of energy arrival and channel fading on the power control policy.
As illustrated in Fig. \ref{fig:mdpmodel}(b),
$E_m$ units of energy arrives at the beginning of the $m$-th epoch. The arrived energy is stored in a battery with capacity $B_{\text{max}}$ before it is used. Without loss of generality, we assume that $E_m \leq B_{\text{max}}$, i.e., the amount of arrived energy is at most $B_{\text{max}}$. There is no initial energy in the battery, i.e., the battery is empty before the first epoch.
For the $n$-th block of epoch $m$, the maximum channel gain of the users is denoted as $\gamma_{m,n}$,  which corresponds to the modulation and coding scheme (MCS) with the highest transmission rate, that the channel gain can support. Note that $\gamma_{m,n}$ is measured when the reference transmit power is 1W.

{
For the scenario with multiple carriers, we assume that all carriers are used for transmission at the same time and have the same transmission power.
If there are $C$ carriers, the channel gain of carrier $c$ in the $n$-th block of epoch $m$ is denoted by $\gamma_{m,n,c}$. The spectrum efficiency is
\begin{align}
\frac{1}{C}\sum_{c=1}^{C}{\log\left(1+\gamma_{m,n,c}p\right)} & = \frac{1}{C}{\log\left(\prod_{c=1}^{C}(1+\gamma_{m,n,c}p)\right)} \\
& ={\log\left(\sqrt[C]{\prod_{c=1}^{C}(1+\gamma_{m,n,c}p)}\right)}
\approx \log\left(1+\sqrt[C]{\prod_{c=1}^{C}\gamma_{m,n,c}}p\right).
\end{align}
where $p$ is the transmission power.
In the optimal power control policy, blocks with good channel states are used for transmission, and the transmission power should be large enough due to the baseband processing power. The values of $\gamma_{m,n,c}p$ should be large, and thus the approximation is accurate. We now get the approximated channel gain in each block, i.e., $\gamma_{m,n} = \sqrt[C]{\prod_{c=1}^{C}\gamma_{m,n,c}}$, and the problem with multiple carries can be approximated as scenarios with single carrier. We thus explore the scenario with single carrier in the remaining part of this paper.
}

\begin{figure}
	\centering
	\includegraphics[width=0.4\textwidth]{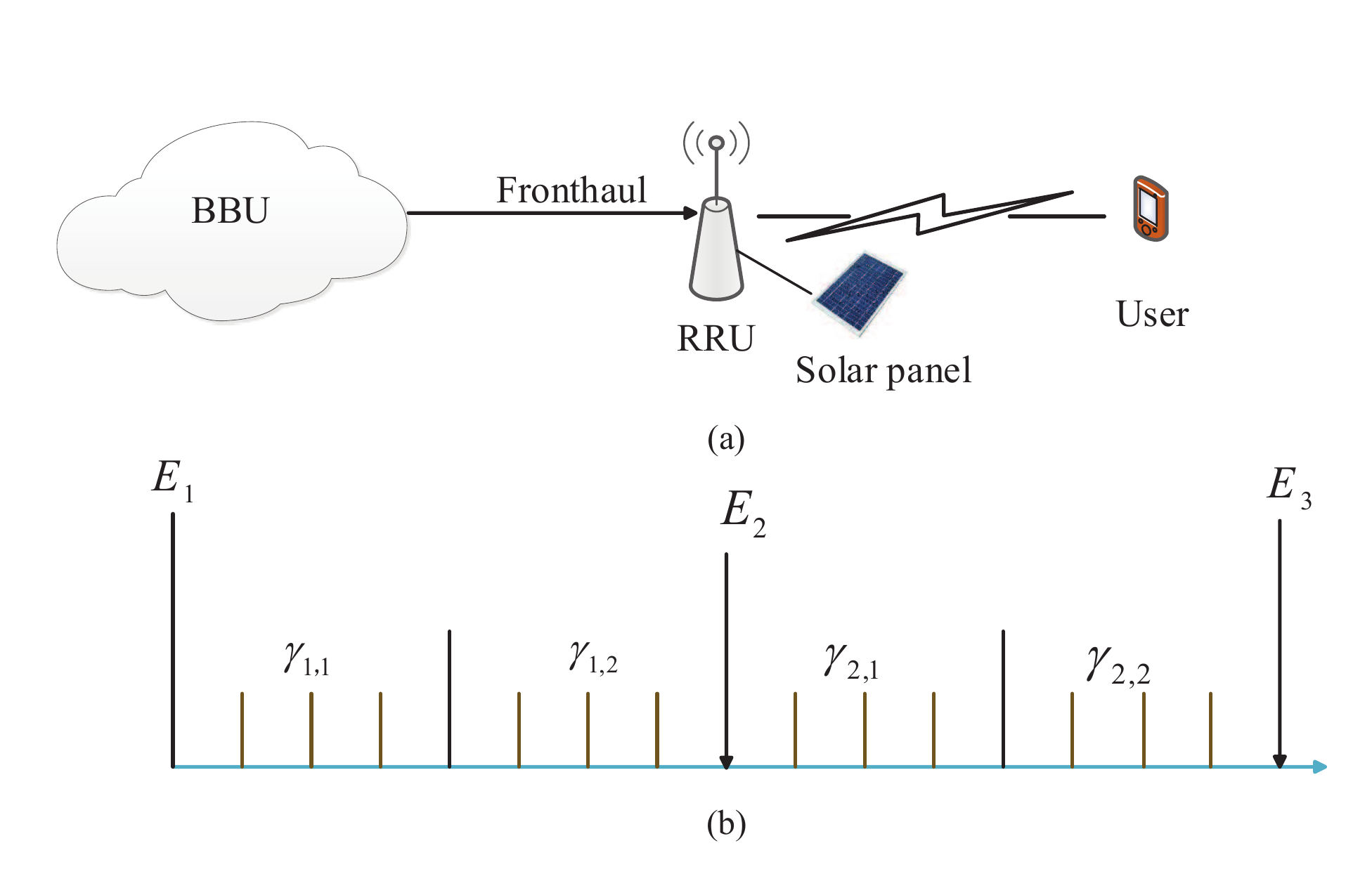}
	\caption{Illustration of C-RAN with renewable energy powered RRU. (a) C-RAN system with renewable energy powered RRU. (b) Energy arrival and channel fading with different time scale.}
	\label{fig:mdpmodel}
\end{figure}

The RRU can be configured with $X$ candidate functional split modes.
In each block, one or more functional split modes can be selected, but at most one functional split mode can be selected at any slot.
In the $n$-th block of epoch $m$, the number of slots that functional split mode $x$ is selected is denoted by $\theta_{m, n, x}$. Note that $\theta_{m,n,x}=0$ means that mode $x$ is not selected in the $n$-th block of epoch $m$.
During one block, the total number of slots used for transmission of the $X$ modes should satisfy
$\sum_{x=1}^{X}\theta_{m,n,x} \leq L$,
where $L$ is number of slots in each block.
The transmission power with mode $x$ in each block should be constant, denoted by $p_{m,n,x}$. The maximum transmission power is $P_{\text{max}}$, i.e., $0 \leq p_{m,n,x} \leq P_{\text{max}}$.
{
The processing power of mode $x$ is denoted by $\varepsilon_k$, and the fronthaul rate requirement is denoted by $R_x$.  The processing power $\varepsilon_k$ and fronthaul rate requirement $R_x$ are related to the MCS \cite{centralize}, and thus related to the transmission power, which makes the problem difficult to analyze. To simplify the problem, we assume that for each functional split mode $x$, the processing power $\varepsilon_k$ and fronthaul rate requirement $R_x$ are constant, which correspond to the MCS with the maximum transmission power $P_{\text{max}}$.
}
Also for the overhead of fronthaul, the average fronthaul rate is constrained to be no more than a given threshold $D$.
As the downlink scenario is considered, the energy consumption of the fronthaul happens at the BBU.
The RRU only consumes energy when it is transmitting data to the users.
In this case, $\theta_{m,n,x}\log (1+ \gamma_{m,n} p_{m,n,x})$ bits of data are transmitted to the users with energy consumption $\theta_{m,n,x}(p_{m,n,x}+\varepsilon_x)$ in the $n$-th block of epoch $m$ with mode $x$.

{
For scenarios with multiple RRUs, if RRUs are self-powered and there is no cooperative transmission, the functional split selection and power control can be done separately at each RRU while treating the signals of other RRUs on the same frequency as noise. As for the scenario with cooperative transmission, we need to further optimize the precoding, and each RRU has its own energy constraints. However, due to the wireless fronthaul implementation and much more complex fronthaul topology, fronthaul sharing and and multiplexing gain should be further considered. Scenarios with cooperative transmission and fronthaul resource management are left as future work.
}

\section{Maximizing the Throughput} \label{sec:maximize}
We consider the offline throughput maximization problem over a finite time of $M$ epochs.
Due to the causality constraints, the energy that has not arrived can not be used, we have
\begin{align}
\sum_{m=1}^{\hat{m}}\sum_{n=1}^{N}\sum_{x=1}^{X}\theta_{m,n,x}(p_{m,n,x}+\varepsilon_x) \leq \sum_{m=1}^{\hat{m}}E_m, \quad \hat{m}=1, 2, ..., M,
\end{align}
note that $\sum_{n=1}^{N}\sum_{x=1}^{X}\theta_{m,n,x}(p_{m,n,x}+\varepsilon_x)$ is the energy consumed in epoch $m$.
There may be energy waste due to the limited battery size when the maximum transmit power is limited, which makes the energy constraints difficult to express. We thus ignore the  maximum transmit power constraint when establishing the offline throughput maximization problem, and then approximate the transmit power that is larger than $P_{\text{max}}$ as $P_{\text{max}}$ in the optimal solution of the problem.
As the energy in the battery at any time can not exceed the battery capacity, at the beginning of epoch $m$, at which time the battery has the most energy in epoch $m$,
there should be
\begin{align}
\sum_{m=1}^{\hat{m}+1}E_m-\sum_{m=1}^{\hat{m}}\sum_{n=1}^{N}\sum_{x=1}^{X}\theta_{m,n,x}(p_{m,n,x}+ \varepsilon_x) \leq B_{\text{max}}, \quad \hat{m}=1, 2, ..., M-1.
\end{align}

Denoted by $\alpha_{m,n,x}=\theta_{m,n,x}p_{m,n,x}$, which is the energy consumed by the radio transmission in the  $n$-th block of epoch $m$ with mode $x$, the optimization problem can be formulated as
\begin{align}
\max_{\theta_{m, n, x}, \alpha_{m,n,x}}  \quad & \sum_{m=1}^{M}\sum_{n=1}^{N}\sum_{x=1}^{X}\theta_{m, n, x} \log(1+\gamma_{m,n} \frac{\alpha_{m,n,x}}{\theta_{m,n,x}}) \label{eq:obj}\\
\text{s.t.} \quad &\frac{1}{MNL}\sum_{m=1}^{M}\sum_{n=1}^{N}\sum_{k=1}^{K}\theta_{m, n,k}R_k \leq D, \label{eq:FHConstraint}\\
&\sum_{m=1}^{\hat{m}}\sum_{n=1}^{N}\sum_{x=1}^{X}(\alpha_{m,n,x}+\varepsilon_x\theta_{m,n,x}) \leq \sum_{m=1}^{\hat{m}}E_m, \quad 1 \leq \hat{m} \leq M\\
&\sum_{m=1}^{\hat{m}+1}E_m-\sum_{m=1}^{\hat{m}}\sum_{n=1}^{N}\sum_{k=1}^{K}(\alpha_{m,n,k}+\varepsilon_k\theta_{m,n,k}) \leq B_{\text{max}}, \quad 1 \leq \hat{m} \leq M-1\\
&\sum_{x=1}^{X}\theta_{m, n,x} \leq L, \quad \forall m,n \label{eq:blockLength}\\
&\alpha_{m,n,x} \geq 0, \theta_{m,n,x} \geq 0, \quad \forall m, n, x
\end{align}
where Eq. (\ref{eq:FHConstraint}) is the  constraint of the average fronthaul rate, and Eq. (\ref{eq:blockLength}) is the constraint of the block length.
Note that the functional split mode is included in the optimization of $\alpha_{m,n,x}$, i.e., $\alpha_{m,n,x}>0$ means that mode $x$ is selected in the $n$-th block of epoch $m$, otherwise mode $x$ is not selected.
Note that we can treat the number of slots $\theta_{m,n,x}$ as a continuous variable in the first place, in which case the complexity of solving the optimization problem can be greatly reduced, and some intuitive results can be given, while at the same time the effect on the throughput is small after approximating $\theta_{m,n,x}$ into an integer  when $L$ is large.
As the optimization objective in Eq. (\ref{eq:obj}) is convex, and the constraints are linear, this is a convex problem. With Lagrangian multiplier method, we are able to get the following structure of the optimal solution.

\begin{prop} \label{prop:1}
{
In the $n$-th block of epoch $m$, during which the channel gain stays constant, the optimal transmission power $p_{m, n, x}$ of the selected modes are the same for any mode $x$ in the optimal solution.
}
\end{prop}

\begin{proof}
The Lagrangian with $\phi \geq 0$, $\mu_{\hat{m}} \geq 0$, $\nu_{\hat{m}} \geq 0$, $\tau_{m,n}\geq0$, $\eta_{m,n,x} \geq 0$ and $ \xi_{m,n,x} \geq 0$ can be written as
\begin{align}
\mathcal{L}=&\sum_{m=1}^{M}\sum_{n=1}^{N}\sum_{x=1}^{X}\theta_{m,n,x} \log(1+\gamma_{m,n} \frac{\alpha_{m,n,x}}{\theta_{m,n,x}}) -\phi\left(\frac{1}{MNL}\sum_{m=1}^{M}\sum_{n=1}^{N}\sum_{x=1}^{X}\theta_{m, n,x}R_x-D \right) \nonumber\\
& -\sum_{\hat{m}=1}^{M} \mu_{\hat{m}} \left[\sum_{m=1}^{\hat{m}}\sum_{n=1}^{N}\sum_{x=1}^{X}(\alpha_{m,n,x}+\varepsilon_x\theta_{m,n,x}) - \sum_{m=1}^{\hat{m}}E_m\right]  \nonumber\\
& -\sum_{\hat{m}=1}^{M-1} \nu_{\hat{m}} \left[\sum_{m=1}^{\hat{m}+1}E_m-\sum_{m=1}^{\hat{m}}\sum_{n=1}^{N}\sum_{x=1}^{X}(\alpha_{m,n,x}+\varepsilon_x\theta_{m,n,x}) - B_{\text{max}} \right] \nonumber \\
& -\sum_{m=1}^{M} \sum_{n=1}^{N} \tau_{m,n} \left(\sum_{x=1}^X\theta_{m, n,x} - L\right) \\ 
& + \sum_{m=1}^{M}\sum_{n=1}^{N}\sum_{x=1}^{X} \eta_{m,n,x} \alpha_{m,n,x} + \sum_{m=1}^{M}\sum_{n=1}^{N} \sum_{x=1}^{X}\xi_{m,n,x}\theta_{m,n,x}
\end{align}
Taking derivatives with respect to $\alpha_{m,n,x}$ and $\theta_{m,n,x}$ , there should be
\begin{align}
\frac{\partial{\mathcal{L}}}{\partial{\alpha_{m,n,x}}}=&\frac{\gamma_{m,n}\theta_{m,n,x}}{\theta_{m,n,x}+\gamma_{m,n}\alpha_{m,n,x}}-\sum_{\hat{m}=m}^{M}\mu_{\hat{m}}+\sum_{\hat{m}=m}^{M-1}\nu_{\hat{m}}
+\eta_{m,n,x}, \label{eq:alpha} \\
\frac{\partial{\mathcal{L}}}{\partial{\theta_{m,n,x}}}=&\log(1+\gamma_{m,n} \frac{\alpha_{m,n,x}}{\theta_{m,n,x}})-\frac{\gamma_{m,n}\alpha_{m,n,x}}{\theta_{m,n,x}+\gamma_{m,n}\alpha_{m,n,x}}-\frac{\phi}{MNL} R_x -\sum_{\hat{m}=m}^{M}\mu_{\hat{m}}\varepsilon_x\nonumber \\
&+\sum_{\hat{m}=m}^{M-1}\nu_{\hat{m}}\varepsilon_x -\tau_{m,n} + \xi_{m,n,x} \label{eq:theta}
\end{align}
If mode $x$ is selected in the $n$-th block of epoch $m$, we have $\alpha_{m,n,x}>0$, with the complementary slackness condition
$\eta_{m,n,x} \alpha_{m,n,x}=0$,
we have $\eta_{m,n,x}=0$. According to (\ref{eq:alpha}), let $\frac{\partial{\mathcal{L}}}{\partial{\alpha_{m,n,x}}}=0$,
\begin{equation}
\frac{\gamma_{m,n}\theta_{m,n,x}}{\theta_{m,n,x}+\gamma_{m,n}\alpha_{m,n,x}}=\sum_{\hat{m}=m}^{M}\mu_{\hat{m}}-\sum_{\hat{m}=m}^{M-1}\nu_{\hat{m}},
\end{equation}
i.e., for $\forall n$ and $\forall x$, the transmit power $p_{m,n,x}$ can be expressed as
\begin{equation} \label{eq:sumeq}
p_{m,n,x}=\frac{1}{\sum_{\hat{m}=m}^{M}\mu_{\hat{m}}-\sum_{\hat{m}=m}^{M-1}\nu_{\hat{m}}} - \frac{1}{\gamma_{m,n}}.
\end{equation}
The values of $p_{m,n,x}$ are the same for any selected mode $x$ in the $n$-th block of epoch $m$.
\end{proof}

Proposition \ref{prop:1} reveals that in one block, the transmission power with different functional modes are the same. Further more, we  can find that the sum of the transmit power $p_{m,n,x}$ and the reciprocal of the channel gain $\frac{1}{\gamma_{m,n}}$ are the same for any selected mode $x$  and block $n$ in epoch $m$ according to Eq. (\ref{eq:sumeq}).

\begin{prop} \label{prop:2}
{
In each epoch, i.e.,  the duration between successive energy arrivals, the optimal functional split mode selection policy satisfies that at most two functional split modes are selected.
}
\end{prop}
\begin{proof}
Denoted by $p_{m,n}^*$ the optimal transmission power in the $n$-th block of epoch $m$, and the corresponding transmission duration with functional split mode $x$ is $\theta_{m,n,x}^*$.
The baseband data amount transmitted via fronthaul in epoch $m$ is defined as $F_m^*$, which can be expressed as
$F_m^*=\sum_{n=1}^{N}\sum_{x=1}^{X}\theta_{m,n,x}^* R_x$.
The number of slots used for transmission in block $n$ is defined as $\theta_{m,n}^{\text{block}}$, i.e.,
$\theta_{m,n}^{\text{block}}=\sum_{x=1}^{X}\theta_{m,n,x}^*$.
Given $\theta_{m,n}^{\text{block}}$ and $p_{m,n}^*$, the throughput and the energy consumed by radio transmissions are fixed. The transmission duration $\theta_{m,n,x}^*$ should be an optimal solution of the following subproblem:
\begin{align}
\min_{\theta_{m,n,x}} \quad &\sum_{n=1}^{N}\sum_{x=1}^{X}\theta_{m,n,x}\epsilon_{x} \nonumber \\
\text{s.t.} \quad &\sum_{n=1}^{N}\sum_{x=1}^{X}\theta_{m,n,x}R_{x} = F_{m}^*, \nonumber \\
&\sum_{x=1}^{X}\theta_{m,n,x}=\theta_{m,n}^{\text{block}}, \quad \forall n \nonumber \\
& \theta_{m,n,x} \geq 0,
\end{align}
where the optimization objective $\sum_{n=1}^{N}\sum_{x=1}^{X}\theta_{m,n,x}\epsilon_{x}$ is the energy consumed by baseband processing,
which means that with transmission duration $\theta_{m,n,x}^*$, the least energy is consumed by baseband processing, i.e., we aim to minimize the energy consumption while guaranteeing the transmission time and average fronthaul rate constraint.
The number of slots used to transmit in epoch $m$ is defined as $\theta_m^{\text{epoch}}$, where
$\theta_m^{\text{epoch}}=\sum_{n=1}^{N}\theta_{m,n}^{\text{block}}$.
 We consider the constraints of the total transmission duration in each epoch, instead of the constraint of the total transmission duration in each block, the subproblem can be relaxed as:
\begin{align}
\min_{\theta_{m,n,x}} \quad &\sum_{n=1}^{N}\sum_{x=1}^{X}\theta_{m,n,x}\epsilon_{x} \nonumber \\
\text{s.t.} \quad &\sum_{n=1}^{N}\sum_{x=1}^{X}\theta_{m,n,x}R_{x} = F_{m}^*, \nonumber \\
&\sum_{n=1}^{N}\sum_{x=1}^{X}\theta_{m,n,x}=\theta_m^{\text{epoch}}, \nonumber \\
& \theta_{m,n,x} \geq 0.
\end{align}

In epoch $m$, the energy consumed by baseband processing and the amount of data transmitted via fronthaul are only related to the total transmission duration of each mode, i.e., $\hat{\theta}_x^{\text{mode}}=\sum_{n=1}^{N}\theta_{m,n,x}$. For any optimal solution of the relaxed subproblem, we can find an equivalent solution of the subproblem, and thus the optimal solution of the subproblem is also the optimal solution of the relaxed subproblem.
The Lagrangian of the relaxed subproblem is
\begin{align}
\mathcal{Z}=&\sum_{n=1}^{N}\sum_{x=1}^{X}\theta_{m,n,x}\epsilon_{x}-\rho \left(\sum_{n=1}^{N}\sum_{x=1}^{X}\theta_{m,n,x}R_{x} - F_{m}^*\right) \nonumber \\
&-\upsilon \left(\sum_{n=1}^{N}\sum_{x=1}^{X}\theta_{m,n,x}-\theta_{m}^*\right)+\sum_{n=1}^{N}\sum_{x=1}^{X}\psi_{n,x} \theta_{m,n,x}
\end{align}
Taking derivatives with respect to $\theta_{m,n,x}$, we have
$\frac{\partial{\mathcal{Z}}}{\partial{\theta_{m,n,x}}}=\epsilon_{x}-\rho R_{x} - \upsilon + \psi_{n,x}$.
If mode $x$ is selected in the $n$-th block of epoch $m$, we have $\theta_{m,n,x}>0$
according to the complementary slackness condition that $\psi_{n,x} \theta_{m,n,x}=0$, we have $\psi_{n,x}=0$.
Let $\frac{\partial{\mathcal{Z}}}{\partial{\theta_{m,n,x}}}=0$, there should be
$\epsilon_{x}-\rho R_{x} - \upsilon=0$.
If  more than two functional split modes are selected, assume that the number of selected functional split modes is $Z$, and the selected modes are $x_z$ for $1 \leq z \leq Z $. The following equations should have solution
\begin{equation} \label{eq:formulation}
\begin{cases}
\epsilon_{x_1}-\rho R_{x_1} & - \upsilon =0 \\
\epsilon_{x_2}-\rho R_{x_2} & - \upsilon =0 \\
&\vdots\\
\epsilon_{x_Z}-\rho R_{x_Z} & - \upsilon =0 \\
\end{cases}
\end{equation}
Note that the formulation (\ref{eq:formulation}) has solution only when $Z \leq 2$, or $R_x$ and $\epsilon_x$ satisfies that
\begin{equation}
\frac{\epsilon_{x_2}-\epsilon_{x_1}}{R_{x_1}-R_{x_2}}=\frac{\epsilon_{x_3}-\epsilon_{x_1}}{R_{x_1}-R_{x_3}},
\end{equation}
for any 3 selected modes, which is a trivial scenario that can be ignored, and thus at most two functional split modes can be selected at each epoch.
\end{proof}

{ The solution obtained with continuous transmission duration is denoted by \emph{`upper bound'}.  We now introduce how to round the `upper bound' into integer transmission duration.
Slots with good channel states are used for transmission.
The number of slots used for transmission with functional split mode $x$ in block $n$ is denoted by $\hat{\theta}_{n,x}^{\text{block}}$, with the corresponding transmission power $p_{m,n,x}$, the energy used for transmission is $E_x^{\text{T}} = \sum_{n=1}^N \hat{\theta}_{n,x}^{\text{block}} p_{m,n,x}$. The energy used for baseband processing is $E_x^{\text{B}} = \sum_{n=1}^N \hat{\theta}_{n,x}^{\text{block}} \epsilon_x$, where $\epsilon_x$ is the baseband processing power.
Number of slots used for transmission of each selected functional split mode is rounded into integer, denoted by $\tilde{\theta}_{n,x}^{\text{block}}$.
Besides the baseband processing energy, i.e., $\tilde{E}_x^{\text{B}} = \sum_{n=1}^N \tilde{\theta}_{n,x}^{\text{block}} \epsilon_x$, the energy used for transmission is $\tilde{E}_x^{\text{T}} = E_x^{\text{B}} + E_x^{\text{T}} - \tilde{E}_x^{\text{B}}$.
The transmission power of each slot after rounding is then calculated according to Proposition \ref{prop:1}, with the constraints of the total transmission energy $\tilde{E}_x^{\text{T}}$.
}

{
According to Proposition \ref{prop:2}, we conclude that at most two functional split modes are selected in one epoch, which means that the functional split mode selection can be determined at the time scale of energy arrival, rather than at the time scale of channel fading.
In this sense, the switching of functional split mode can be done in a large time scale. The switching of functional split mode can be implemented by activating and deactivating functions in RRUs and BBU when RRUs and BBU are constructed by using container technologies, the introduced delay (less than
millisecond \cite{sleepfs}) and energy can be neglected.
}

\section{Optimal Online Policy} \label{sec:mdp}
For the online policy, only the causal (past and present ) energy states and channel states are known at the RRU.
To find the optimal online policy, we formulate the online problem as an MDP.
The channel gain varies at the beginning of each block, and each block has $L$ slots. The beginning of the $(n+1)$-th block is the $(nL+1)$-th slot.
The channel gain is modeled as a Markov chain with $G$ states, and the channel gain of state $g$ is $\Gamma_g$. 
The transition probability from state $g_1$ to state $g_2$ at the beginning of the $(n+1)$-th block is denoted as $p_{g_1,g_2}=\text{Pr}\{\gamma_{nL+1}=\Gamma_{g_2}|\gamma_{nL}=\Gamma_{g_1}\}$.

The energy arrives once an epoch. We assume that the energy arrives at the beginning of each epoch. An epoch has $N$ blocks, i.e., $NL$ slots.
The energy arrival is modeled as a finite state Markov chain with $E_{\text{max}}$ states, and the arrived energy amount with state $e$ is $A_e$.  
The transition probability from state $e_1$ to state $e_2$ at the beginning of the $(m+1)$-th epoch is $q_{e_1,e_2}=\text{Pr}\{E_{mNL+1}=A_{e_2}|E_{(m-1)NL+1}=A_{e_1}\}$.  
The arrived energy is stored in a battery with capacity $B_{\text{max}}$ before it is used.
The transmission power in slot $k$ is denoted as $P_k$. Denoted by $x_k$ the functional split mode selected in slot $k$, the baseband processing power is $\epsilon_{x_k}$.

The energy is consumed only when the RRU transmits data to the users, i.e., when $P_k>0$, the state of energy in the battery $B_k$ is updated as
\begin{subequations}
\begin{numcases} {B_{k+1}=}
\min\{B_k+E_k-\epsilon_{x_k}-P_k, B_{\text{max}}\}, &$P_k>0$ \nonumber \\
\min\{B_k+E_k, B_{\text{max}}\}, &$P_k=0$ \nonumber
\end{numcases}
\end{subequations}
To simplify the expression, we introduce a new variable, defined as
\begin{subequations}
\begin{numcases} {\delta_k=}
1, &$P_k>0$\\
0, &$P_k=0$
\end{numcases}
\end{subequations}
then the battery state is updated as
\begin{align}
B_{k+1}=\min\{B_k+E_k-\delta_k(\epsilon_{x_k}+P_k), B_{\text{max}}\}.
\end{align}

The system state is
\begin{align}
s_k=(B_k, E_k, Y_k, \gamma_k, n_k, l_k),
\end{align}
where $B_k$ is the state of energy available in the battery, $E_k$ is the energy arrived in stage $k$, $Y_k$ records the energy arrival rate of the current epoch, $\gamma_k$ is the channel gain, $n_k$ indicates how many blocks the current epoch has lasted, $l_k$ indicates how many slots the current block has lasted.
The state transition probability is
\begin{align}
\text{Pr}\{s_{k+1}|s_k,P_k,x_k\}= & \text{Pr}\{B_{k+1}|B_k, P_k, E_k,x_k\}\text{Pr}\{E_{k+1}|Y_k,E_k,n_k,l_k\} \times \nonumber\\
&\text{Pr}\{Y_{k+1}|Y_k,E_k,n_k,l_k\}\text{Pr}\{n_{k+1}|n_k,l_k\}\text{Pr}\{\gamma_{k+1}|\gamma_k,l_k\}\text{Pr}\{l_{k+1}|l_k\}
\end{align}
{
The value of $n_k$ varies at the beginning of each block, and the state transition is described in Fig. \ref{fig:trans}(a). } The transition probability of $n_k$ is expressed as
 \begin{subequations}
\begin{numcases} {\text{Pr}\{n_{k+1}|n_k,l_k\}=}
1, &$\text{if} \quad n_{k+1}=\text{mod}(n_k, N)+1, l_k = L$ \nonumber\\
1, &$\text{if} \quad n_{k+1}=n_k, l_k < L$ \nonumber\\
0, &$\text{else}$ \nonumber
\end{numcases}
\end{subequations}
where $\text{mod}(n_k, N)$ is modulus operation which returns the remainder after division of $n_k$ by $N$.
The value of $l_k$ varies at the beginning of each slot. The transition probability of $l_k$ is expressed as
 \begin{subequations}
\begin{numcases} {\text{Pr}\{l_{k+1}|l_k\}=}
1, &$\text{if} \quad l_{k+1}=\text{mod}(l_k, L)+1$ \nonumber\\
0, &$\text{else}$ \nonumber
\end{numcases}
\end{subequations}

\begin{figure}
	\centering
	\includegraphics[width=0.75\textwidth]{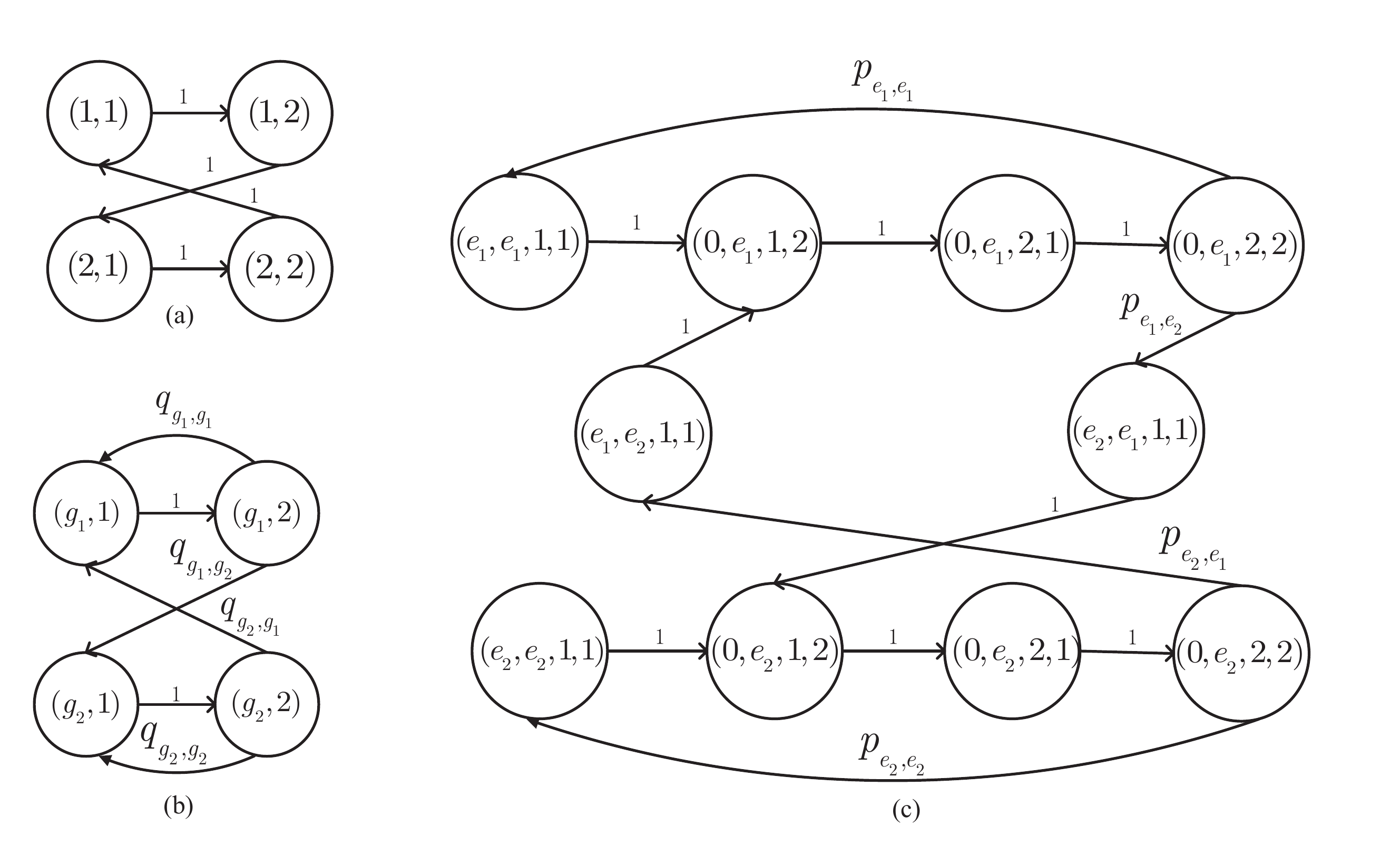}

	\caption{{Illustration of state transitions: (a) state transition of ($n_k$, $l_k$) when $N$=2 and $L$=2;
(b) state transition of ($\gamma_k$, $l_k$) when G=2 and $L$=2;
(c) state transition of ($E_k$, $Y_k$, $n_k$, $l_k$) when $E_{\text{max}}$=2, $N$=2 and $L$=2.}} \label{fig:trans}
\end{figure}

{
The state transition of channel state is described in Fig. \ref{fig:trans}(b), and
the transmission probability of $\gamma_k$ is
}
\begin{subequations}
\begin{numcases} {\text{Pr}\{\gamma_{k+1}|\gamma_k,l_k\}=}
q_{g_1,g_2}, &$\text{if} \quad \gamma_{k+1}=\Gamma_{g_2}, \gamma_k=\Gamma_{g_1}, l_k=L$ \nonumber \\
1, &$\text{if} \quad \gamma_{k+1}=\gamma_k, l_k<L$ \nonumber \\
0, &$\text{else}$ \nonumber
\end{numcases}
\end{subequations}

The transition probability of battery state is
\begin{subequations}
\begin{numcases}{\text{Pr}\{B_{k+1}|B_k, P_k, E_k,x_k\}=}
1, &$\text{if} \quad B_{k+1}=\min\{B_k+E_k-\delta_k(P_k+\epsilon_{x_k}), B_{\text{max}}\}$ \nonumber \\
0, &$\text{else}$ \nonumber
\end{numcases}
\end{subequations}
As energy arrives every $N$ blocks, the energy only arrives at the beginning of each epoch.
{
The state transition of energy arrival state is described in Fig. \ref{fig:trans}(c), and
the transmission probability of $E_k$ is expresses as
}
\begin{subequations}
\begin{numcases} {\text{Pr}\{E_{k+1}|Y_k,n_k,l_k\}=}
1, &$\text{if} \quad E_{k+1}=0, n_k<N$ \nonumber \\
1, &$\text{if} \quad E_{k+1}=0, l_k<L$ \nonumber \\
p_{e_1,e_2}, &$\text{if} \quad E_{k+1}=A_{e_2}, Y_k=A_{e_1}, n_k=N, l_k=L$ \nonumber \\
0, &$\text{else}$ \nonumber
\end{numcases}
\end{subequations}
The value of $Y_k$ only changes after a new instance of energy arrival. The transmission probability of $Y_k$ is expressed as
 \begin{subequations}
\begin{numcases} {\text{Pr}\{Y_{k+1}|Y_k,E_k,n_k,l_k\}=}
1, &$\text{if} \quad n_k>1, Y_{k+1}=Y_k$ \nonumber\\
1, &$\text{if} \quad l_k>1, Y_{k+1}=Y_k$ \nonumber\\
1, &$\text{else if} \quad n_k=1, l_k=1, Y_{k+1}=E_k$ \nonumber \\
0, &$\text{else}$ \nonumber
\end{numcases}
\end{subequations}

Due to the constraints of the energy in the battery and the maximum transmit power, the transmit power should be constrained as
\begin{align}
0 \leq P_k \leq \min\{ B_k-\epsilon_{x_k}, P_{\text{max}}\},
\end{align}
where $P_{\text{max}}$ is the maximum allowed transmission power.
According to Shannon's equation, denoted by
\begin{align}
r(P_k,\gamma_k)=\log(1+\gamma_k P_k).
\end{align}
{
The objective function is set as
\begin{align}
\lim_{K \to \infty}  \max_{P_k, x_k} \frac{1}{K}E\left[\sum_{k=1}^{K}r(P_k,\gamma_k)-\eta\sum_{k=1}^{K}\delta_kR_{x_k}\right],
\end{align}
where $r(P_k,\gamma_k)$ is the transmission rate of stage $k$ given the transmission power $P_k$ and the channel gain $\gamma_k$, which corresponds to the throughput in stage $k$, the expectation is taken over the channel gain and the energy arrival rate; $\delta_kR_{x_k}$ is the amount of baseband signals transmitted via fronthaul in slot $k$, i.e., the fronthaul overhead, which corresponds to the average fronthaul rate.
The optimization variable is the transmission power and the functional split mode selection, and $\eta$ is a weighting factor. We can tradeoff between the throughput and the fronthaul overhead by adjusting $\eta$. With large $\eta$, we have stringent constraint on the average fronthaul rate. To satisfy a given constraint of average fronthaul rate, we can iterate the weighting factor $\eta$ with algorithms such as the gradient descent method \cite{iter}.
}

The average throughput maximization problem is formulated as an MDP, and the value iteration algorithm can be used to find the optimal policy \cite{DP}. Every slot is treated as a stage. Denoted by $a_k=\{P_k, x_k\}$ the action taken in stage $k$. The reward function in stage $k$ is denoted by
\begin{equation}
g(s_k,a_k)=\log(1+\gamma_k P_k)-\eta \delta_kR_{x_k},
\end{equation}
The objective is to minimize the average per-stage reward of the infinite horizon problem, which is denoted by
\begin{equation}
J^*(s_0)=\lim_{K \to \infty}\max_{\pi} \frac{1}{K}E\left[\sum_{k=0}^{K-1}g(s_k,a_k)\right], \label{eq:mdp}
\end{equation}
where $s_0$ is the initial state, $\pi=\{a_0, a_1, ..., a_{K-1}\}$ is the possible policy.
Problem (\ref{eq:mdp}) can be solved with value iteration algorithm. Denoted by $\lambda$ the average per-stage reward, $h(i)$ the relative reward when starting at state $i$, the Bellman equation is expressed as
\begin{equation}
\lambda+h(i)=\max_{a}\{ g(i, a)+ \sum_{j \in \mathcal{S}}\text{Pr}\{j|i,a\}h(j)\},
\end{equation}
where $\mathcal{S}$ is the set of all possible states.
Initialize $h^{(0)}(i)=0$. Given any state $s$, for the $(b+1)$-th iteration, we have
\begin{align}
h^{(b+1)}(i)= \max_{a}\{ g(i, a)+ \sum_{j \in \mathcal{S}}\text{Pr}\{j|i,a\}h^{(b)}(j)\}-\max_{a}\{ g(s, a)+ \sum_{j \in \mathcal{S}}\text{Pr}\{j|s,a\}h^{(b)}(j)\},
\end{align}
note that $\max_{a}\{ g(s, a)+ \sum_{j \in \mathcal{S}}\text{Pr}\{j|s,a\}h^{(b)}(j)\}$ converges to $\lambda$.
A more general iteration formulation is
\begin{align}
h^{(b+1)}(i)= &(1-\tau) h^{(b)}(i)+\max_{a}\{ g(i, a)+ \tau \sum_{j \in \mathcal{S}}\text{Pr}\{j|i,a\}h^{(b)}(j)\} \nonumber\\
&-\max_{a}\{ g(s, a)+ \tau \sum_{j \in \mathcal{S}}\text{Pr}\{j|s,a\}h^{(b)}(j)\},
\end{align}
where $0 < \tau < 1$. Denote the gap between  $h^{(b+1)}(i)$ and $h^{(b)}(i)$ as $d^{(b)}(i)$, i.e.,
\begin{align}
d^{(b)}(i)=h^{(b+1)}(i)-h^{(b)}(i).
\end{align}
The iteration is considered as convergence when
\begin{align}
\max_{i}d^{(b)}(i)-\min_{i}d^{(b)}(i)<\omega,
\end{align}
where $\omega$ is a threshold which determines the convergence speed. The detailed value iteration algorithm is described in Algorithm \ref{alg:mdp}.
\begin{algorithm}
\caption{Value Iteration Algorithm}
\label{alg:mdp}
\begin{algorithmic}
\STATE {Initialize $b=0$, $h^{(0)}(i)=0$ for $\forall i \in \mathcal{S}$, $\lambda^{(0)}=0$, select a fixed state $s_0$ }
\REPEAT
\STATE{1. Update the average per-stage reward $\lambda$}
\STATE{\begin{equation}
\lambda^{(b+1)}=\max_{a}\{ g(s_0, a) +  \tau \sum_{j \in \mathcal{S}}\text{Pr}\{j|s_0,a\}h^{(b)}(j)\}\}. \nonumber
\end{equation}
}
\STATE{2. Update $h$}
\STATE{
\begin{align}
h^{(b+1)}(i)= (1-\tau) h^{(b)}(i) + \max_{a}\{ g(i, a)+ \tau \sum_{j \in \mathcal{S}}\text{Pr}\{j|i,a\}h^{(b)}(j)\}-\lambda^{(b+1)} \nonumber
\end{align}
}
\STATE{3. Update $b=b+1$}
\UNTIL{$\max_{i}d^{(b)}(i)-\min_{i}d^{(b)}(i)<\omega$}
\end{algorithmic}
\end{algorithm}

For the optimal online problem, the state number of the MDP model is $(B_{\text{max}}+1)\times E_{\text{max}} \times E_{\text{max}} \times G \times N \times L$,  and the number of actions is $(P_{\text{max}}+1)\times X$. The state space can be very large if some of the elements is of large size. The value iteration algorithm may encounter curse of dimensionality. In this case, lower-complexity algorithm is in need. In the next section, we will first analyze the power control policy with one instance of energy arrival, based on which a heuristic online algorithm is proposed.

\section{Single Energy Arrival, Constant Channel Gain} \label{sec:single}
According to Proposition \ref{prop:2}, at most two functional split modes are selected in each epoch in the optimal offline problem.
To gain some insights, we will give some intuitive results when there is only one instance of energy arrival, and the channel gain is constant, i.e., $M=1$, $X=2$. Note that if the channel gain is averaged over an epoch, one epoch can be approximated to have only one block, where the approximated channel gain is the average channel gain over the epoch.
For brevity, we will use $\theta_1$, $\theta_2$, $p_1$, $p_2$ instead of $\theta_{1,1,1}$, $\theta_{1,1,2}$, $p_{1,1,1}$ and $p_{1,1,2}$ in this section, i.e., $\theta_1$ and $\theta_2$ are the corresponding transmission durations with the 2 functional split modes, $p_1$ and $p_2$ are the transmission power, the amount of available energy in this epoch is denoted by $E$, the epoch length is denoted by $L$. If there are more than two candidate functional split modes, i.e., $X>2$, we can first calculate the throughput when any two of the functional split modes are selected (there are totally $\frac{X(X-1)}{2}$ possible combinations), and obtain the optimal power control policy by comparing the throughput of all the possible scenarios.

If only one mode is selected, denoted by mode $x$, the optimal power control policy can be obtained with ``glue pouring" \cite{glue}. Given the processing power $\varepsilon_j$ and channel gain $\gamma$, and without maximum transmission duration constraint, the throughput maximization problem can be simplified to
\begin{equation}
\max_{p_x} \frac{E}{p_x+\varepsilon_x}\log(1+\gamma p_x),
\end{equation}
where $p_x$ is the transmission power, and denote $v_x^*$ as the optimal transmission power obtained by solving the optimization problem.
The optimal transmission power $v_x^*$ satisfies:
\begin{align}
(1+\gamma v_x^*)\log(1+\gamma v_x^*)-\gamma v_x^*=\gamma \varepsilon_x. \label{eq:glue}
\end{align}
Note that the expression on the left side of the equality is an increasing function of $v_x^*$, the equation has an unique solution, and $v_x^*$ increases with $\varepsilon_x$. Due to the constraints of epoch length and average fronthaul rate, the transmission duration is limited. Denoted by $\theta_x^{\text{max}}=\min\{\frac{DL}{R_x}, L\}$, which is the maximum transmission duration when only mode $x$ is selected. When $E<\theta_x^{\text{max}}(v_x^*+\varepsilon_x)$, the optimal power control policy is
\begin{equation}
p_x=v_x^*,\quad  \theta_x=\frac{E}{v_x^*+\varepsilon_x}.
\end{equation}
When $E \geq \theta_x^{\text{max}}(v_x^*+\varepsilon_x)$, the optimal power control policy is
\begin{equation}
p_x=\frac{E}{\theta_x^{\text{max}}}-\varepsilon_x, \quad \theta_x=\theta_x^{\text{max}}.
\end{equation}

Due to the average fronthaul rate constraint $D$, the power control policy is affected. We will derive the optimal power control policy under different values of $D$ in the following part of this section. We assume that the two modes  are mode 1 and mode 2, where mode 1 has less baseband functions at the RRU, and we thus have $R_1 > R_2$, $\epsilon_1  < \epsilon_2$.

\subsection{$D \geq R_1$}

\begin{figure}
	\centering
	\includegraphics[width=0.45\textwidth]{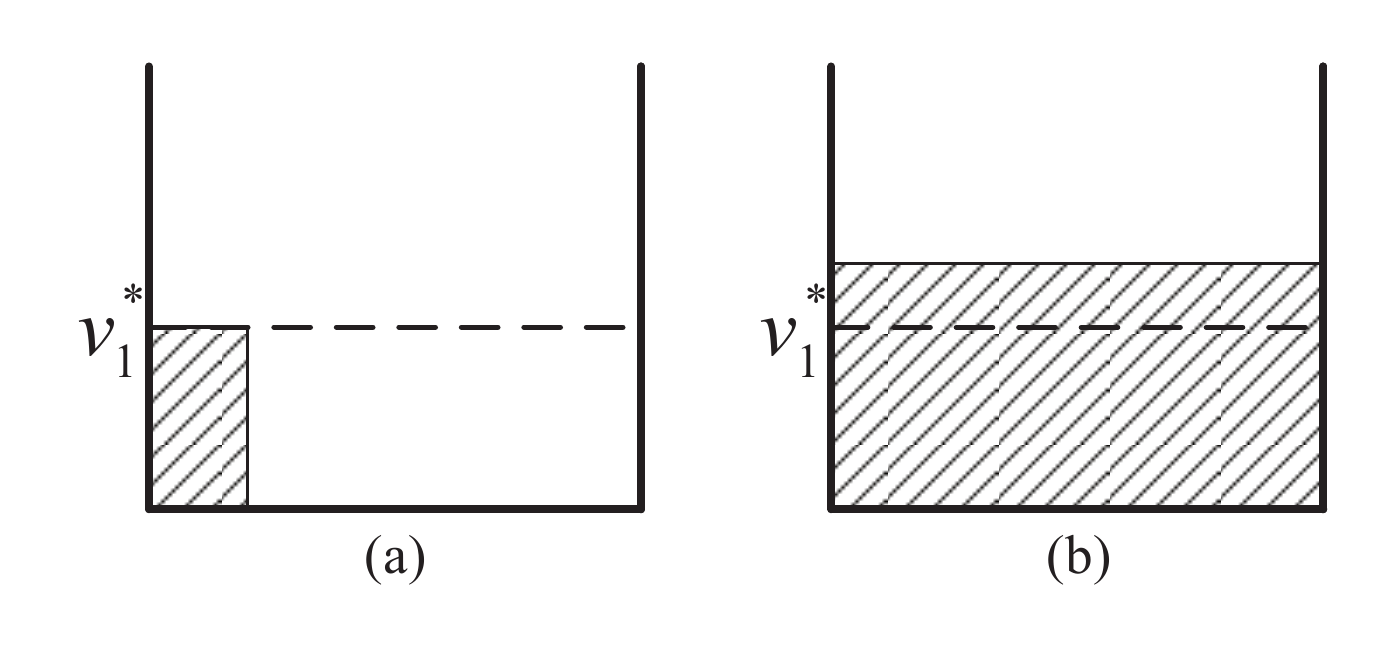}
	\caption{The optimal power control policy when $D \geq R_1$, where $\theta_1$ and $p_1$ are represented by the width and height of the black shadowed block with up diagonal respectively: (a) $E<(v_1^*+\varepsilon_1)L$; (b) $E\geq(v_1^*+\varepsilon_1)L$. }
	\label{fig:case1}
\end{figure}

When $D \geq R_1$, the average fronthaul rate constraint can always be satisfied, and thus only mode $1$, which has smaller processing power, is selected.
When $E<(v_1^*+\varepsilon_1)L$, the optimal power control policy is
\begin{align}
\theta_1=\frac{E}{v_1^*+\varepsilon_1}, \quad p_1=v_1^*, \quad \theta_2=0,  \quad p_2=0,
\end{align}
as described in Fig. \ref{fig:case1}(a).
When $E\geq(v_1^*+\varepsilon_1)L$, the optimal power control policy is
\begin{align}
\theta_1=L,  \quad p_1=\frac{E}{L}-\varepsilon_1,  \quad \theta_2=0,  \quad p_2=0,
\end{align}
as described in Fig. \ref{fig:case1}(b).

\subsection{$R_2 < D < R_1$} \label{sec:sub2}

\begin{figure}
	\centering
	\includegraphics[width=0.45\textwidth]{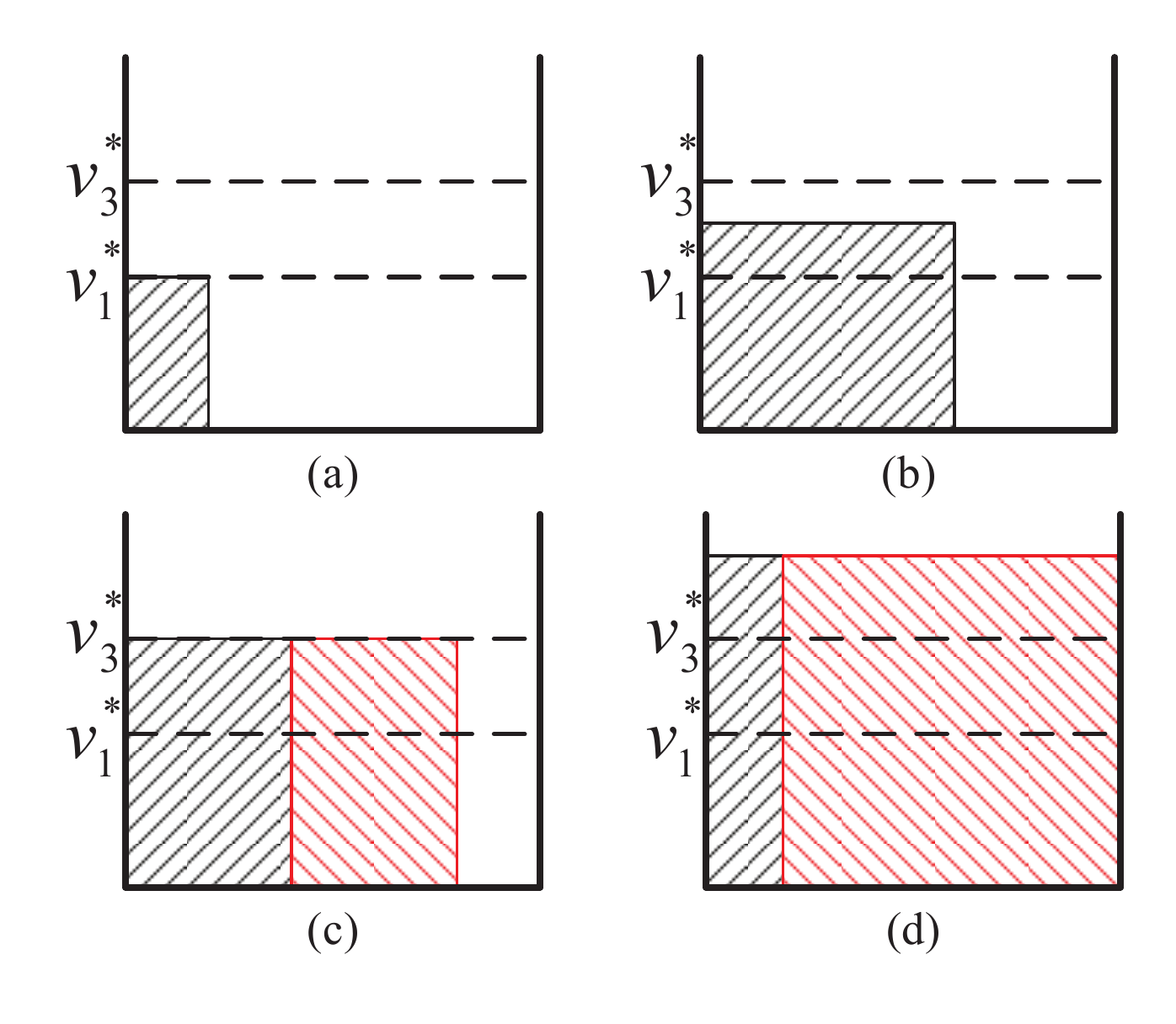}
	\caption{The optimal power control policy when $R_2 < D < R_1$, $\theta_1$ and $p_1$ are represented by the width and height of the black shadowed block with up diagonal, respectively, $\theta_2$ and $p_2$ are represented by the width and height of the red shadowed  block with down diagonal, respectively: (a) $E \leq \frac{DL(v_1^*+\varepsilon_1)}{R_1}$; (b) $\frac{DL(v_1^*+\varepsilon_1)}{R_1} < E \leq \frac{DL(v_3^*+\varepsilon_1)}{R_1}$; (c) $\frac{DL(v_3^*+\varepsilon_1)}{R_1} < E \leq Lv_3^*+\frac{DL(\varepsilon_1-\varepsilon_2)+(R_1\varepsilon_2-R_2\varepsilon_1)L}{R_1-R_2}$; (d) $E>Lv_3^*+\frac{DL(\varepsilon_1-\varepsilon_2)+(R_1\varepsilon_2-R_2\varepsilon_1)L}{R_1-R_2}$.}
	\label{fig:case2}
\end{figure}

When $E \leq \frac{DL(v_1^*+\varepsilon_1)}{R_1}$, if only functional split mode 1 is selected with transmission power $v_1^*$, the transmission time is $\frac{E}{v_1^*+\varepsilon_1}$, where the average fronthaul rate constraint can be satisfied.
Thus the optimal power control policy is
\begin{align}
\theta_1=\frac{E}{v_1^*+\varepsilon_1},  \quad p_1=v_1^*,  \quad \theta_2=0,  \quad p_2=0,
\end{align}
as described in Fig. \ref{fig:case2}(a).

If $\theta_1R_1+\theta_2R_2<DL$, i.e., the amount of data transmitted via fronthaul is less than the allowed amount $DL$, functional split mode 2, which has larger processing power, should not be selected.
When $D<R_1$, and $E \geq \frac{DL(v_1^*+\varepsilon_1)}{R_1}$, if only functional split mode 1 is selected, $\theta_1=\frac{DL}{R_1}$, we have $\theta_1R_1+\theta_2R_2=DL$. We can draw the conclusion that when $E \geq \frac{DL(v_1^*+\varepsilon_1)}{R_1}$, there should be
$\theta_1R_1+\theta_2R_2=DL$.
According to Proposition \ref{prop:1}, the transmission power of the two modes are the same, denoted by $p$, and thus we have
$\theta_1(p+\varepsilon_1)+\theta_2(p+\varepsilon_2)=E$,
the transmission duration can be expressed as
\begin{align}
\theta_1=\frac{(p+\varepsilon_2)DL-R_2E}{R_1(p+\varepsilon_2)-R_2(p+\varepsilon_1)}, \quad \theta_2=\frac{R_1E-(p+\varepsilon_1)DL}{R_1(p+\varepsilon_2)-R_2(p+\varepsilon_1)}.
\end{align}
The throughput is
\begin{align}
H=\frac{(R_1-R_2)E+(\varepsilon_2-\varepsilon_1)DL}{R_1(p+\varepsilon_2)-R_2(p+\varepsilon_1)}\log(1+\gamma p).
\end{align}
Taking the derivative of $H$ with respect to $p$, we have
\begin{align}
\frac{\partial H}{\partial p}=&\frac{(R_1-R_2)\left[(R_1-R_2)E+(\varepsilon_2-\varepsilon_1)DL\right]}{\left[(R_1-R_2)p+R_1\varepsilon_2-R_2\varepsilon_1 \right]^2} \nonumber \\
& \times \left[\frac{\gamma(p+\varepsilon_2+\frac{R_2(\varepsilon_2-\varepsilon_1)}{R_1-R_2})}{1+\gamma p}-\log(1+\gamma p)\right]
\end{align}
Let $\frac{\partial H}{\partial p}=0$, we have
\begin{align} \label{eq:v3}
\frac{\gamma(v_3^*+\varepsilon_2+\frac{R_2(\varepsilon_2-\varepsilon_1)}{R_1-R_2})}{1+\gamma v_3^*}-\log(1+\gamma v_3^*)=0,
\end{align}
this equation is equivalent to (\ref{eq:glue}), which obtains the optimal transmission power in glue pouring, denote by  $v_3^*$. Since $\frac{R_2(\varepsilon_2-\varepsilon_1)}{R_1-R_2}>0$ and $\varepsilon_2>\varepsilon_1$, we have $v_3^*>v_1^*$.

When $p<v_3^*$, we have $\frac{\partial H}{\partial p}>0$, the throughput increases with $p$. The transmission power $p$ should be as large as possible, while satisfying that $\theta_1 \geq 0$ and $\theta_2 \geq 0$.
When $\frac{DL(v_1^*+\varepsilon_1)}{R_1} < E \leq \frac{DL(v_3^*+\varepsilon_1)}{R_1}$, the maximum transmission power $p=\frac{ER_1}{DL}-\varepsilon_1$ is achieved when $\theta_1=\frac{DL}{R_1}$ and $\theta_2=0$,
i.e., the optimal power control policy is
\begin{align}
\theta_1=\frac{DL}{R_1}, \quad p_1=\frac{ER_1}{DL}-\varepsilon_1, \quad \theta_2=0, \quad p_2=0,
\end{align}
i.e., only functional split mode 1 is selected, the transmission power increases with $E$, while the transmission duration remains unchanged, as described in Fig. \ref{fig:case2}(b).

When $p>v_3^*$, we have $\frac{\partial H}{\partial p}<0$, the throughput decreases with $p$.
If $\frac{DL(v_3^*+\varepsilon_1)}{R_1} < E \leq Lv_3^*+\frac{DL(\varepsilon_1-\varepsilon_2)+(R_1\varepsilon_2-R_2\varepsilon_1)L}{R_1-R_2}$, the transmission power can be $v_3^*$, and the transmission duration can be obtained by solving the  following equations:
\begin{align}
\theta_1R_1+\theta_2R_2=DL,  \quad \theta_1(v_3^*+\varepsilon_1)+\theta_2(v_3^*+\varepsilon_2)=E.
\end{align}
The optimal power control policy is
\begin{align}
\theta_1=\frac{(v_3^*+\varepsilon_2)DL-R_2E}{R_1(v_3^*+\varepsilon_2)-R_2(v_3^*+\varepsilon_1)}, \quad
\theta_2=\frac{R_1E-(v_3^*+\varepsilon_1)DL}{R_1(v_3^*+\varepsilon_2)-R_2(v_3^*+\varepsilon_1)}, \quad p=v_3^*, \label{eq:case23}
\end{align}
as described in Fig. \ref{fig:case2}(c). With the increasing of $E$, the transmission power remains unchanged, the transmission duration of functional split mode 1 decreases while  the transmission duration of functional split mode 2 increases.
Note that when $E=Lv_3^*+\frac{DL(\varepsilon_1-\varepsilon_2)+(R_1\varepsilon_2-R_2\varepsilon_1)L}{R_1-R_2}$, the total transmission duration is equal to the epoch length, i.e., $\theta_1+\theta_2=L$.

When $E>Lv_3^*+\frac{DL(\varepsilon_1-\varepsilon_2)+(R_1\varepsilon_2-R_2\varepsilon_1)L}{R_1-R_2}$, due to the epoch length constraint, we have $p>v_3^*$, and the transmission durations of the two functional split modes should satisfy
\begin{align}
\theta_1+\theta_2=L, \quad \theta_1R_1+\theta_2R_2=DL,
\end{align}
i.e., $\theta_1=\frac{DL-R_2L}{R_1-R_2}$, $\theta_2=\frac{R_1L-DL}{R_1-R_2}$. As there is no energy waste, we have
$\theta_1(p+\varepsilon_1)+\theta_2(p+\varepsilon_2)=E$,
i.e., the optimal power control policy is
\begin{align}
\theta_1=\frac{DL-R_2L}{R_1-R_2}, \quad \theta_2=\frac{R_1L-DL}{R_1-R_2}, \quad p=\frac{E}{L}-\frac{D(\varepsilon_1-\varepsilon_2)}{R_1-R_2}-\frac{R_1\varepsilon_2-R_2\varepsilon_1}{R_1-R_2},
\end{align}
as described in Fig. \ref{fig:case2}(d). With the increasing of $E$, the transmission durations of both functional split modes stay unchanged, while the transmission power increases.

\subsection{$D \leq R_2$}
When $D \leq R_2$, the derivation of the optimal transmission power control policy is similar to the analysis in Section \ref{sec:sub2}.
\begin{figure}
	\centering
	\includegraphics[width=0.45\textwidth]{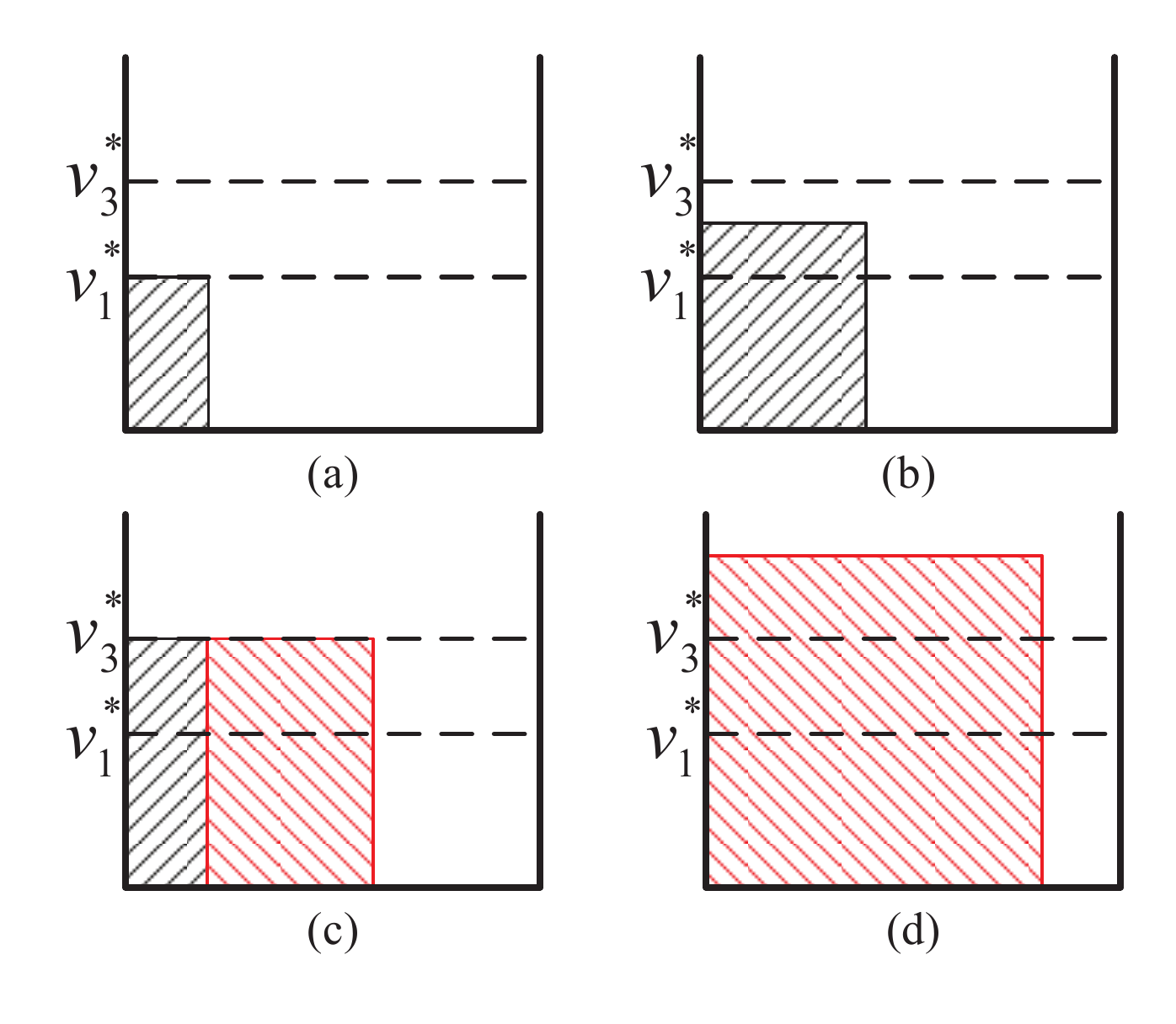}
	\caption{The optimal power control policy when $D \leq R_2$, $\theta_1$ and $p_1$ are represented by the width and height of the black shadowed  block with up diagonal, respectively, $\theta_2$ and $p_2$ are represented by the width and height of the red shadowed block with down diagonal, respectively: (a) $E\leq\frac{DL(v_1^*+\varepsilon_1)}{R_1}$; (b) $\frac{DL(v_1^*+\varepsilon_1)}{R_1} < E \leq \frac{DL(v_3^*+\varepsilon_1)}{R_1}$; (c) $\frac{DL(v_3^*+\varepsilon_1)}{R_1} < E \leq \frac{DL(v_3^*+\varepsilon_2)}{R_2}$; (d) $E>\frac{DL(v_3^*+\varepsilon_2)}{R_2}$. }
	\label{fig:case3}
\end{figure}

When $E<\frac{DL(v_1^*+\varepsilon_1)}{R_1}$, the optimal power control policy is
\begin{align}
\theta_1=\frac{E}{v_1^*+\varepsilon_1}, \quad p_1=v_1^*, \quad \theta_2=0, \quad p_2=0,
\end{align}
as described in Fig. \ref{fig:case3}(a).

When $\frac{DL(v_1^*+\varepsilon_1)}{R_1} < E \leq \frac{DL(v_3^*+\varepsilon_1)}{R_1}$, the optimal power control policy is
\begin{align}
\theta_1=\frac{DL}{R_1}, \quad p_1=\frac{ER_1}{DL}-\varepsilon_1, \quad \theta_2=0, \quad p_2=0,
\end{align}
as described in Fig. \ref{fig:case3}(b).

When $\frac{DL(v_3^*+\varepsilon_1)}{R_1} < E \leq \frac{DL(v_3^*+\varepsilon_2)}{R_2}$, the optimal transmission power $v_3^*$ can be achieved, and the optimal power control policy is
\begin{align}
\theta_1=\frac{(v_3^*+\varepsilon_2)DL-R_2E}{R_1(v_3^*+\varepsilon_2)-R_2(v_3^*+\varepsilon_1)}, \quad \theta_2=\frac{R_1E-(v_3^*+\varepsilon_1)DL}{R_1(v_3^*+\varepsilon_2)-R_2(v_3^*+\varepsilon_1)}, \quad p=v_3^*,
\end{align}
as described in Fig. \ref{fig:case3}(c).

When $E>\frac{DL(v_3^*+\varepsilon_2)}{R_2} $, due to the average fronthaul rate constraint, the transmission duration is limited, we have $p>v_3^*$. As the throughput $H$ decreases with $p$, the transmission power $p$ should be as small as possible. Functional split mode $2$, which has smaller fronthaul rate requirement is selected, and the optimal power control policy is
\begin{align}
\theta_1=0, \quad p_1=0, \quad \theta_2=\frac{DL}{R_2}, \quad p_2=\frac{ER_2}{DL}-\varepsilon_2,
\end{align}
as described in Fig. \ref{fig:case3}(d).

\section{Heuristic Online Policy} \label{sec:online}
{In each block, if the available amount of energy in the battery $E$, the transmission deadline $T$, the channel gain $\gamma$ and the average fronthaul rate $D$ are accurately estimated, the optimal transmission policy can be easily obtained, according to the analyses in Section \ref{sec:single}.
Due to the energy constraint, only blocks with good channel states are used for transmission to improve the energy efficiency. On the other hand, to avoid energy waste introduced by the limited battery size, we prefer to use all energy in the battery before the end of each epoch. In fact, because only blocks with good channel states are used for transmission, the energy can flow to the next epoch if there is no block with good channel states in the current epoch, which guarantees that energy is used in blocks with good channel states.}

To simplify the expression, we define a function $f$ as:
\begin{equation} \label{fun:single}
[\boldsymbol{\theta}, \boldsymbol{p}] = f(E,T,\gamma,D),
\end{equation}
where $\boldsymbol{\theta}=[\theta_1, \theta_2, \cdot \cdot \cdot, \theta_X]$, $\theta_x$ is the optimal transmission duration that functional split mode $x$ is selected, $\boldsymbol{p}=[p_1,p_2,\cdot \cdot \cdot ,p_X]$, and $p_x$ is the optimal transmission power when functional split mode $x$ is selected. We then propose a low-complexity heuristic online algorithm.

\begin{algorithm}
\caption{Heuristic Online Policy}
\label{alg:heuristic}
\begin{algorithmic}
\STATE {Initialize $B_0=0$, $D_0=0$}
\FOR{$n=1, 2, ..., $}
\STATE{1. Update the energy state $B_{(n-1)L+1}=B_{(n-1)L}+E_{(n-1)L+1}$}
\STATE{2. Get the expected number of blocks with good channel states in the current epoch $N_{\text{good}}$, and the average channel gain $\gamma_{\text{avg}}$}
\STATE{3. Update the average fronthaul rate constraint}
\STATE{\begin{align}
d_n=\frac{(n+n_{\text{heu}})D-D_{n-1}}{N_{\text{heu}}} \nonumber
\end{align}}
\STATE{4. Get the power control policy, with $B_{(n-1)L+1}$, $N_{\text{good}}$, $\gamma_{\text{avg}}$, $d_n$ and thus the value of the function in (\ref{fun:single})}
\STATE{5. Transmit with the power control policy $\bf{\theta}$ and $\bf{p}$, update the battery state and the cumulative amount of data transmitted via fronthaul}
\STATE{\begin{align}
&B_{nL} = B_{(n-1)L+1} - \sum_{x=1}^{X}\theta_x(p_x+\epsilon_x), \quad D_n = D_{n-1} + \sum_{x=1}^{X}\theta_xR_x \nonumber
\end{align}}
\ENDFOR
\end{algorithmic}
\end{algorithm}

The detail algorithm is described in Algorithm \ref{alg:heuristic}.  
We evaluate the transmission policy at the beginning of each block, denoted by block $n$ without losing generality. The RRU transmits only when the channel state is good, i.e., the channel gain is larger than a threshold denoted by $\gamma_{\text{th}}$.

At the first step,
evaluate the amount of energy in the battery $B_{(n-1)L+1}$, which can be obtained with the remaining energy in the battery at the end of the last block $B_{(n-1)L}$, and the arrived energy in this block, $E_{(n-1)L+1}$. We have $B_{(n-1)L+1}=B_{(n-1)L}+E_{(n-1)L+1}$.

At the second step, update the expected number of blocks in the remaining time of the epoch with good channel states, denoted by $N_{\text{good}}$, and the average channel gain, denoted by $\gamma_{\text{avg}}$, which can be obtained according to the distribution of the channel gain.
{
Denoted by $\text{Pr}\{n, n_v, w\}$ the probability that there are $n_v$ blocks with state $v$ in the next $n$ blocks, and the channel state in the $n$-th block is $w$.
We have
\begin{align}
\text{Pr}\{n+1, n_v, y\} = \sum_{w}\text{Pr}\{n, n_v, w\}q_{w,y}, \quad y \neq v, \\
\text{Pr}\{n+1, n_v+1, v\} = \sum_{w}\text{Pr}\{n, n_v, w\}q_{w,v}.
\end{align}
Given the channel state $w$ of the first block, if $w = v$, we have $\text{Pr}\{1, n_v = 1, w\} = 1$, and $\text{Pr}\{1, n_v , y \} = 0$ for the other parameters; if $w \neq v$, we have $\text{Pr}\{1, n_v = 0, w\} = 1$, and $\text{Pr}\{1, n_v , y \} = 0$ for the other parameters.
With the initial iteration values and the iterative formula, we can get the distribution of the number of blocks with channel state $v$, i.e., $\sum_{w}\text{Pr}\{N_{\text{r}}, n_v, w\}$, in the remaining $N_{\text{r}}$ blocks.
With distributions of the number of blocks with each channel state, the expected number of blocks with good channel states can be easily obtained.
}

At the third step, update the average fronthaul rate constraint, denoted by $d_n$. We guarantee that the data amount transmitted via fronthaul does not exceed $nLD$ from block 1 to block $n$.
Denote the total amount of transmitted data in the first $n$ blocks as $D_n$. To guarantee that the average fronthaul rate constraint in the first $n+n_{\text{heu}}$ blocks, the amount of data transmitted via fronthaul in the next $n_{\text{heu}}$ blocks should not exceed $(n+n_{\text{heu}})D-D_{n-1}$, where $n_{\text{heu}}$ is a constant number of blocks used in the heuristic algorithm. In the next $n_{\text{heu}}$ blocks, the expected number of blocks with good channel states is denoted as $N_{\text{heu}}$. The average fronthaul rate constraint is estimated as $d_n=\frac{(n+n_{\text{heu}})D-D_{n-1}}{N_{\text{heu}}}$.

At the fourth step, get the power control policy including the transmission duration $\boldsymbol{\theta}$ and the transmission power $\boldsymbol{p}$, with function (\ref{fun:single}).

At the fifth step, transmit with the policy, update the energy in the battery at the end of the block, and the amount of data transmitted via fronthaul in the first $n$ blocks, i.e.,
\begin{align}
B_{nL} = B_{(n-1)L+1} - \sum_{x=1}^{X}\theta_x(p_x+\epsilon_x), \quad D_n = D_{n-1} + \sum_{x=1}^{X}\theta_xR_x.
\end{align}

\section{Numerical Results} \label{sec:num}
{
We consider the downlink transmission of  an energy harvesting RRU, where the baseband processing is according to the LTE protocol, as shown in  Fig. \ref{fig:functionalsplit}.
The baseband functions at the RRU and the BBU are realized with general purpose processors via function virtualization.
Three candidate functional split modes are considered, including: mode 1, which splits between RF and IFFT, is the classical CPRI functional split; mode 2, which splits between resource mapping and precoding, is a reference split by eCPRI \cite{eCPRI}; mode 3, which splits between RLC and PDCP, is a reference split by 3GPP.
The RRU has one antenna, one carrier component, the bandwidth of the air interface is set as 20MHz, and the sampling rate is 30.72 MHz. We assume that there is only 1 user per TTI, occupying all PRBs. The highest modulation order is 64QAM. When the RRU works in different functional split modes,
the corresponding required fronthaul rates are $R_1=983$Mbps, $R_2=466$Mbps and $R_3=151$Mbps, respectively \cite{SplittingBS}. The corresponding processing powers of the RRU are $\varepsilon_1=2$W, $\varepsilon_2=4$W, $\varepsilon_3=5$W respectively, according to the downlink power model in \cite{LTEmodel}.
}

{
We assume that each slot lasts 10 seconds, and set that each block has $L=4$ slots, each epoch has $N=2$ blocks.
The RRU is powered with renewable energy.
The harvested energy can be used after being stored in the battery with capacity $B_{\text{max}}=1000$J. The initially stored energy in the battery is 0 J.
Without loss of generality, we assume that the energy arrival is a Poisson process, which can be used to model the solar panel or wind generation \cite{energymodel}, The distribution of the amount of arrived energy in each epoch is
\begin{equation}
\text{Pr}\{E_m=ANL\}=\frac{E_{{\text{avg}}}^A}{A!}e^{-E_{\text{avg}}},
\end{equation}
where $E_{\text{avg}}$ is the average energy arrival rate normalized by the number of slots in each epoch.
}

We consider the channel gain between each user and the RRU follows Rayleigh channel distribution, with average channel gain $\gamma_\text{avg}=2$/W. The channel gain is discrete into $G=4$ consecutive intervals without overlapping, and the probability that the channel gain is in each interval is $\frac{1}{G}$. The channel gain in each interval is represented by its average value, denoted the channel gain in the $g$-th interval by $\gamma_g$. Assume the channel gain of different users are i.i.d., the best channel gain of the users in each block $\gamma_{\text{best}}$  follows:
\begin{equation}
\text{Pr}\{\gamma_{\text{best}} = \gamma_g \}=\left( \frac{g}{G} \right)^U - \left( \frac{g-1}{G} \right)^U, \quad 1 \leq g \leq G,
\end{equation}
where $U$ is the number of users. We consider the scenario where $U=2$, and the corresponding probability of each interval is $[\frac{1}{16}, \frac{3}{16},\frac{5}{16},\frac{7}{16}]$.

We first study the offline throughput maximization problem.
{
The solution obtained with continuous variables is denoted by the \emph{`upper bound'}, and the solution after rounding is denoted by \emph{`relax and round'}.
}
The relationship between the throughput and the average fronthaul rate is presented in Fig. \ref{fig:offline} and Fig.\ref{fig:online} with `relax and round' and optimal online policy when $E_{{\text{avg}}}=5$W, respectively. We can see that the throughput grows rapidly with the average fronthaul rate when the fronthaul rate is small, and the growth slows down when the fronthaul rate gets large. When the average fronthaul rate is small, the performance of fixing functional split mode as mode 3 can achieve similar performance with the flexible functional split, because fronthaul is the main constraint in this scenario. When the average fronthaul rate is large, fixing functional split mode as mode 1 can achieve close performance with the flexible functional split, because the energy is the main constraint in this scenario, and functional split mode 1 requires the lowest processing power.

{Throughput of the `upper bound', the `relax and round', optimal online policy and heuristic online policy are  compared in Fig. \ref{fig:compare}. We can see that the `relax and round' has close performance with the `upper bound', which means that ``relax and round'' performs very close to the optimal solution. } The heuristic online policy has similar performance with the optimal online policy. The heuristic online policies with fixed functional split mode are adopted as baselines to show the benefit of flexible functional split, as shown in Fig. \ref{fig:heuristic}. We can see that with flexible functional split, the heuristic online policy have better performance than any fixed functional split mode.
%

\begin{figure}[H]
	\begin{minipage}[t]{0.45\textwidth}
		\centering
		\includegraphics[scale=0.5]{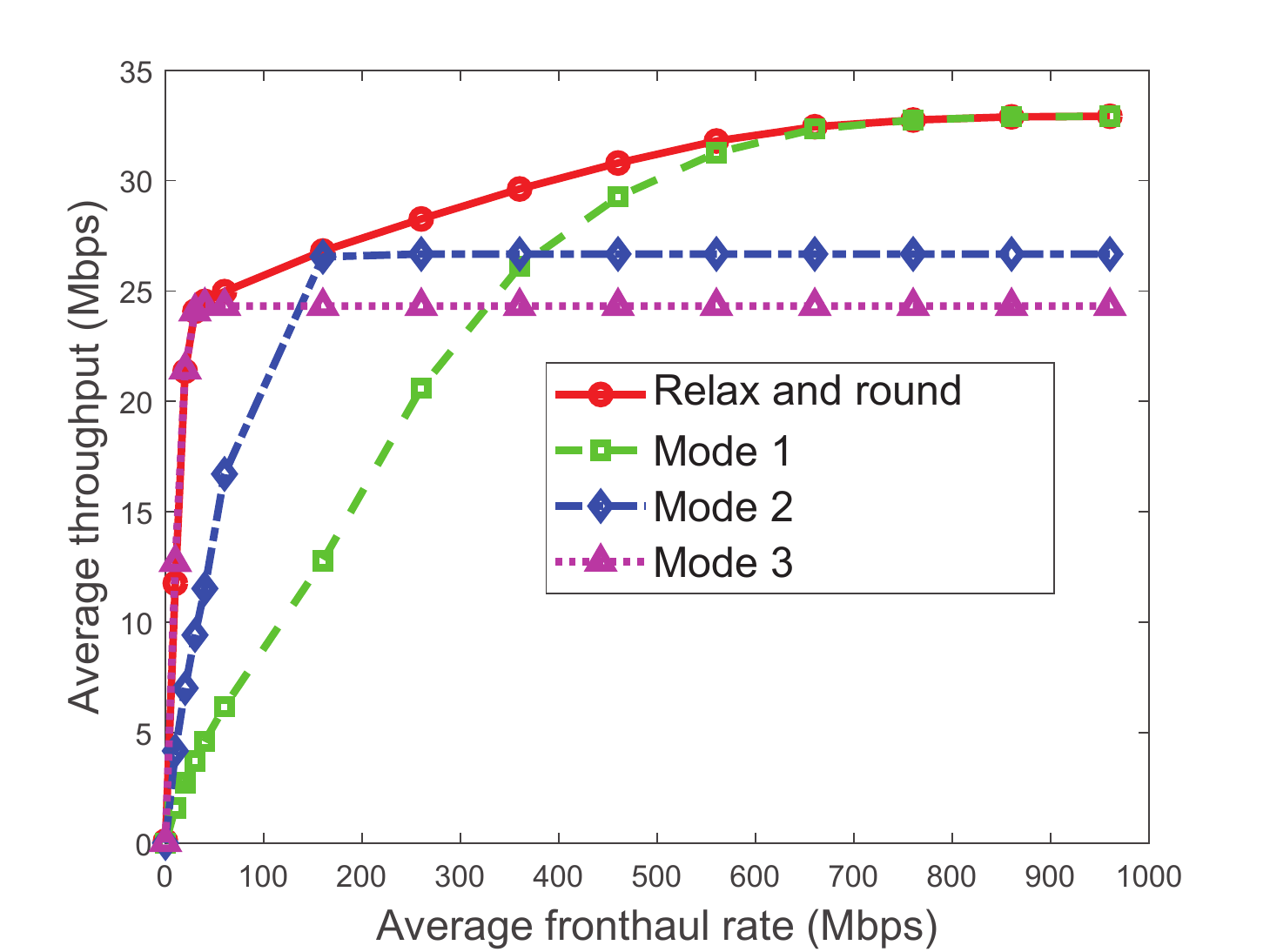}
	\caption{{Comparison of the `relax and round' with flexible functional split and fixed functional split under different average fronthaul rate constraints.}}
	\label{fig:offline}
	\end{minipage}
	\qquad
	\begin{minipage}[t]{0.45\textwidth}
		\centering
		\includegraphics[scale=0.5]{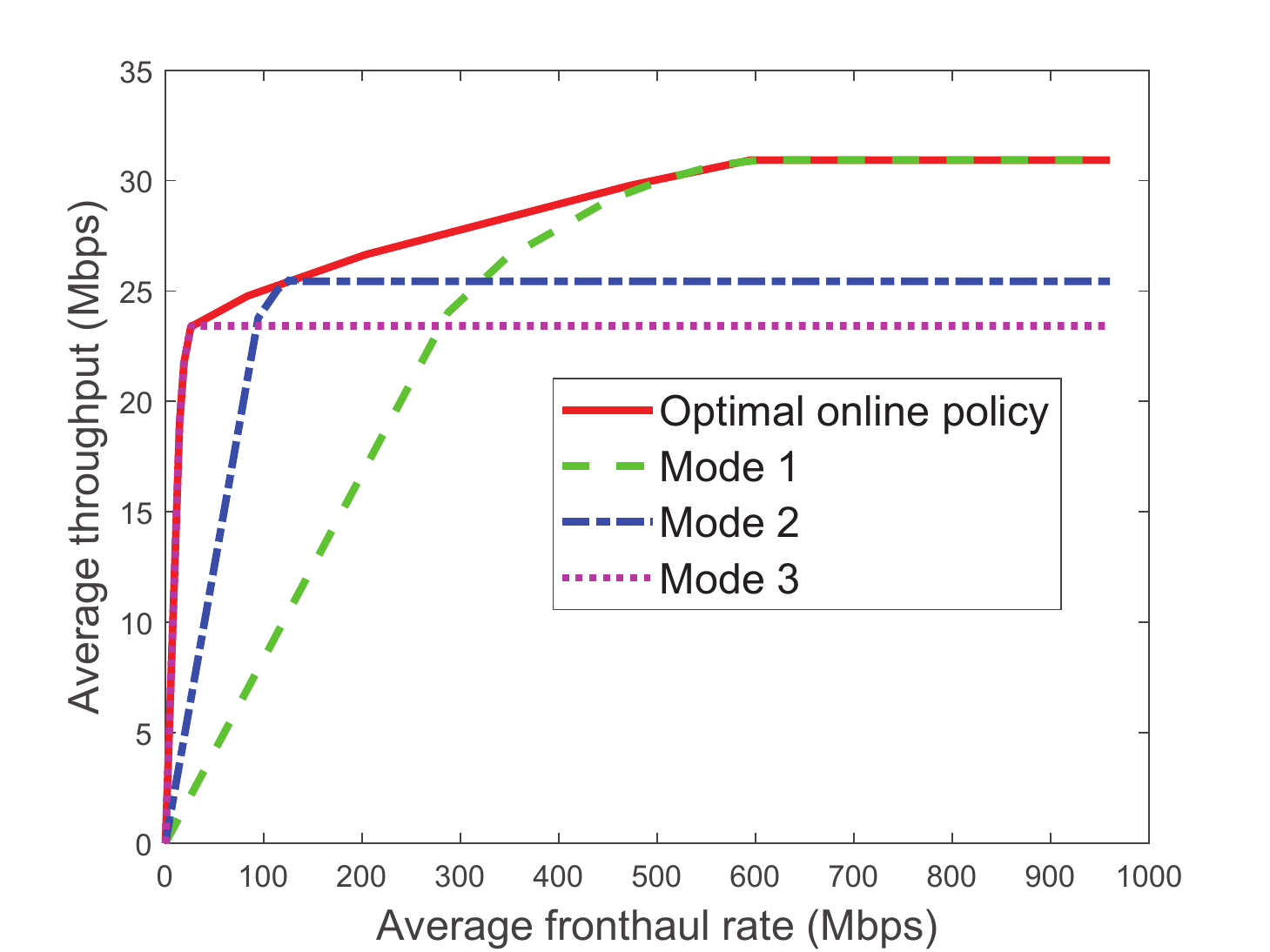}
	\caption{Comparison of the online policies with flexible functional split and fixed functional split under different average fronthaul rate constraints.}
	\label{fig:online}
	\end{minipage}
\end{figure}

\begin{figure}[H]
	\begin{minipage}[t]{0.45\textwidth}
		\centering
		\includegraphics[scale=0.5]{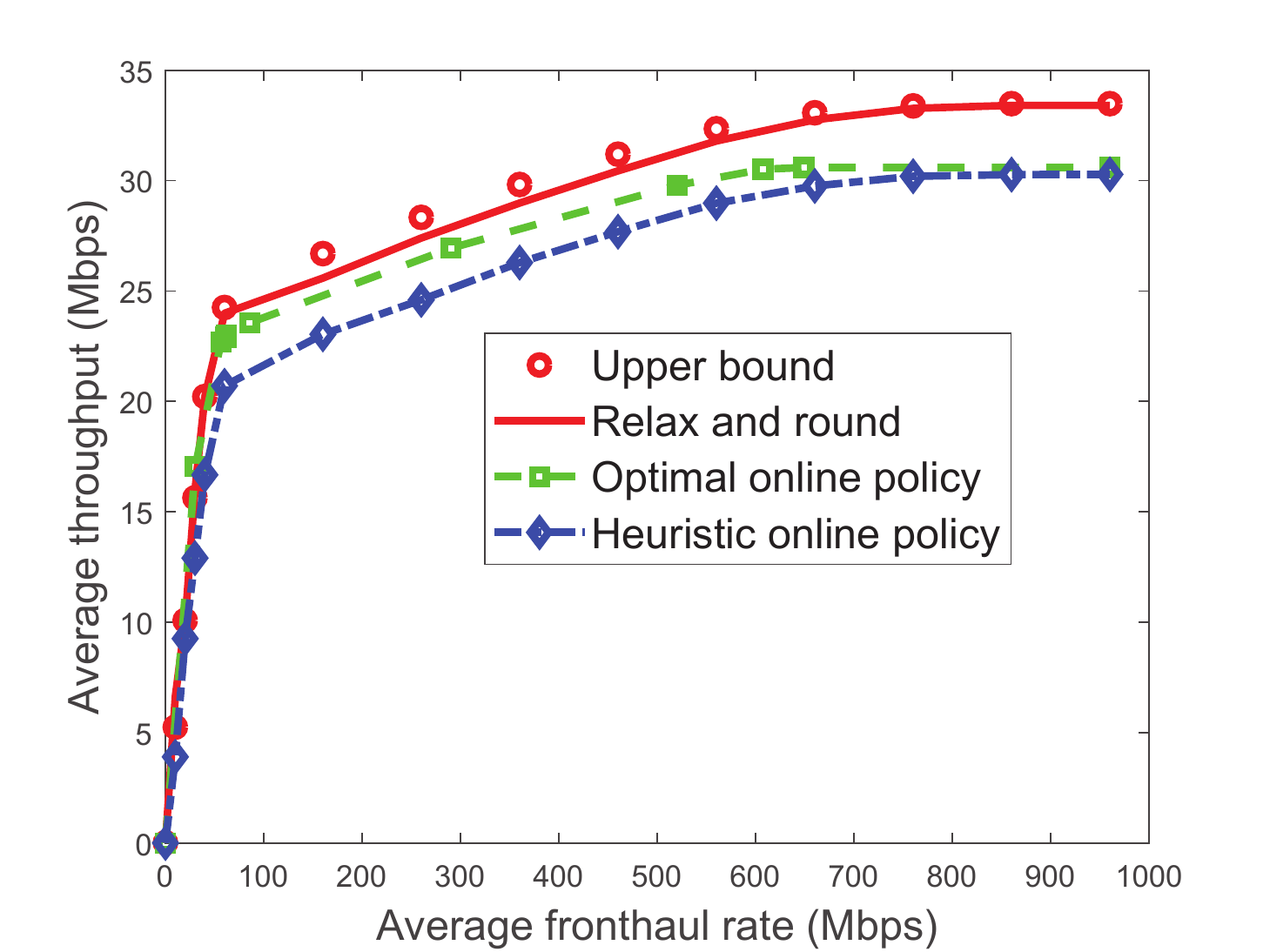}
	\caption{{Comparison of the `upper bound', the `relax and round', optimal online policy and heuristic online policy under different average fronthaul rate constraints.}}
	\label{fig:compare}
	\end{minipage}
	\qquad
	\begin{minipage}[t]{0.45\textwidth}
		\centering
		\includegraphics[scale=0.5]{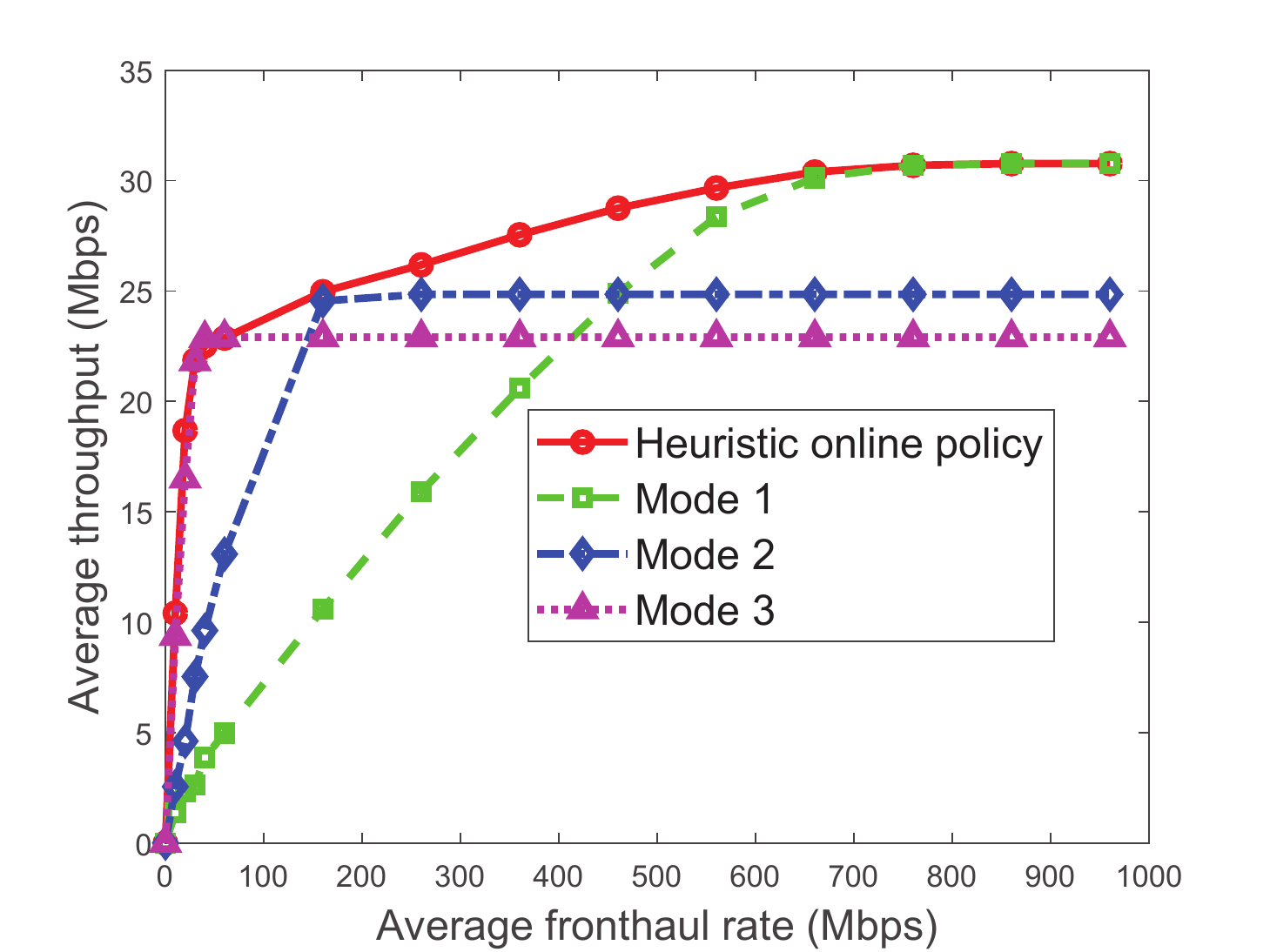}
	\caption{Comparison of the heuristic online policies with flexible functional split and fixed functional split under different average fronthaul rate constraints.}
	\label{fig:heuristic}
	\end{minipage}
\end{figure}

%

To show the effects of the energy arrival rate $E_{\text{avg}}$, the relationship between the average throughput and the energy arrival rate is given in Fig. \ref{fig:energy_compare}, where the average fronthaul rate constraint is 360Mbps. The throughput increases with the energy arrival rate, for both the flexible functional split and the functional split with fixed mode. When the energy arrival is small, the throughput increases approximate linearly with the energy arrival rate. Because in this scenario, the time used for transmission is short, and the energy is the main constraint, rather than the average fonthaul rate and the channel states. However, due to the constraints of average fronthaul rate, functional split modes with smaller fronthaul rate requirement should be selected if longer transmission time is needed, which means the processing power is larger. On the other hand, the number of blocks with good channel states is limited, which means that the available transmission time with good channel states is limited. When the energy arrival rate is large, the increasing of throughput slows down due to the average fronthaul rate constraint and the limited number of blocks with good channel states. With flexible functional split, the throughput is larger compared with the fixed functional split modes.

%

\begin{figure}[H]
	\begin{minipage}[t]{0.45\textwidth}
		\centering
		\includegraphics[scale=0.5]{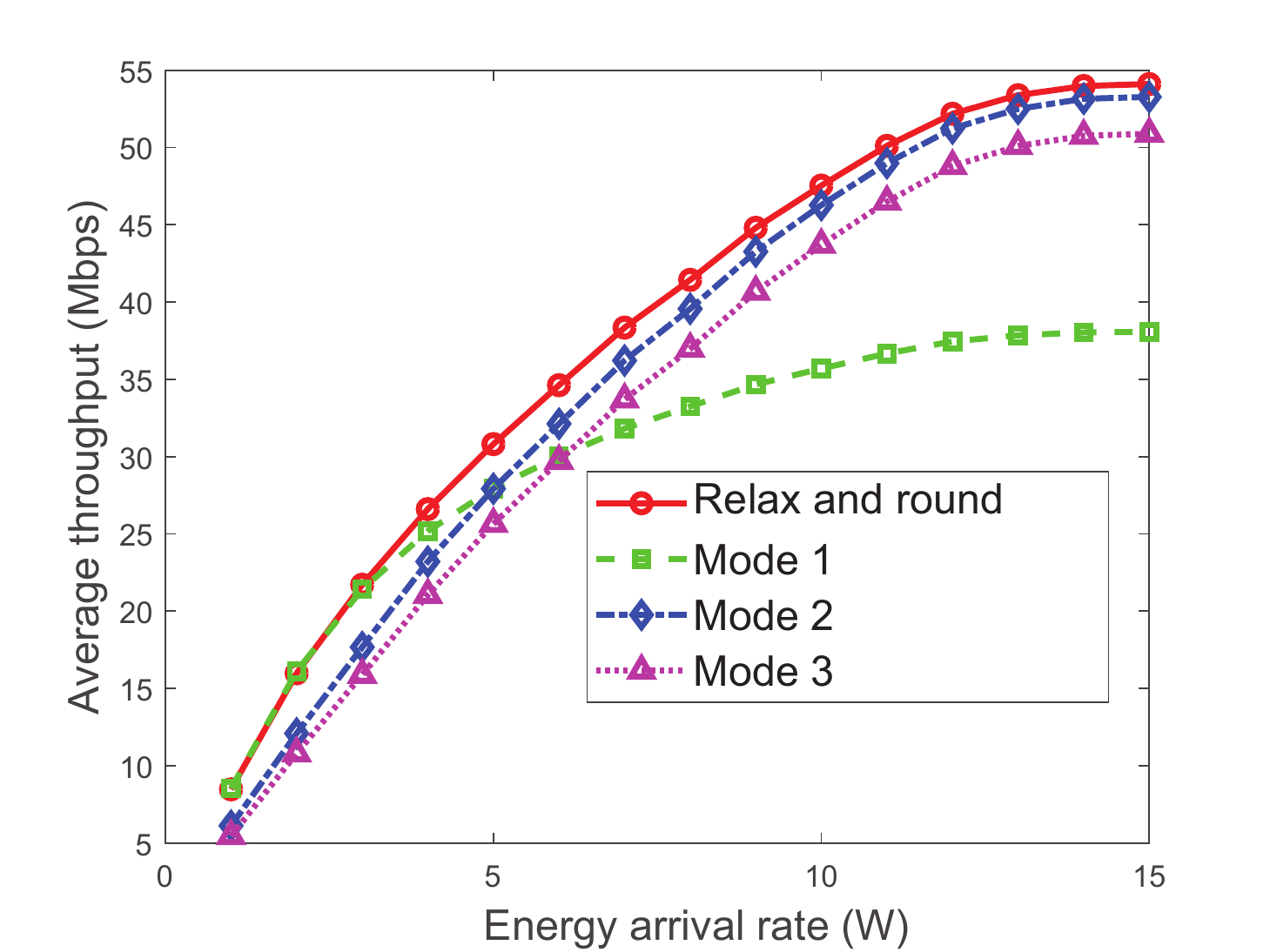}
	   \caption{{Comparison of the `relax and round' with flexible functional split and fixed functional split under different energy arrival rate.}}
	   \label{fig:energy_compare}
	\end{minipage}
	\qquad
	\begin{minipage}[t]{0.45\textwidth}
		\centering
		\includegraphics[scale=0.5]{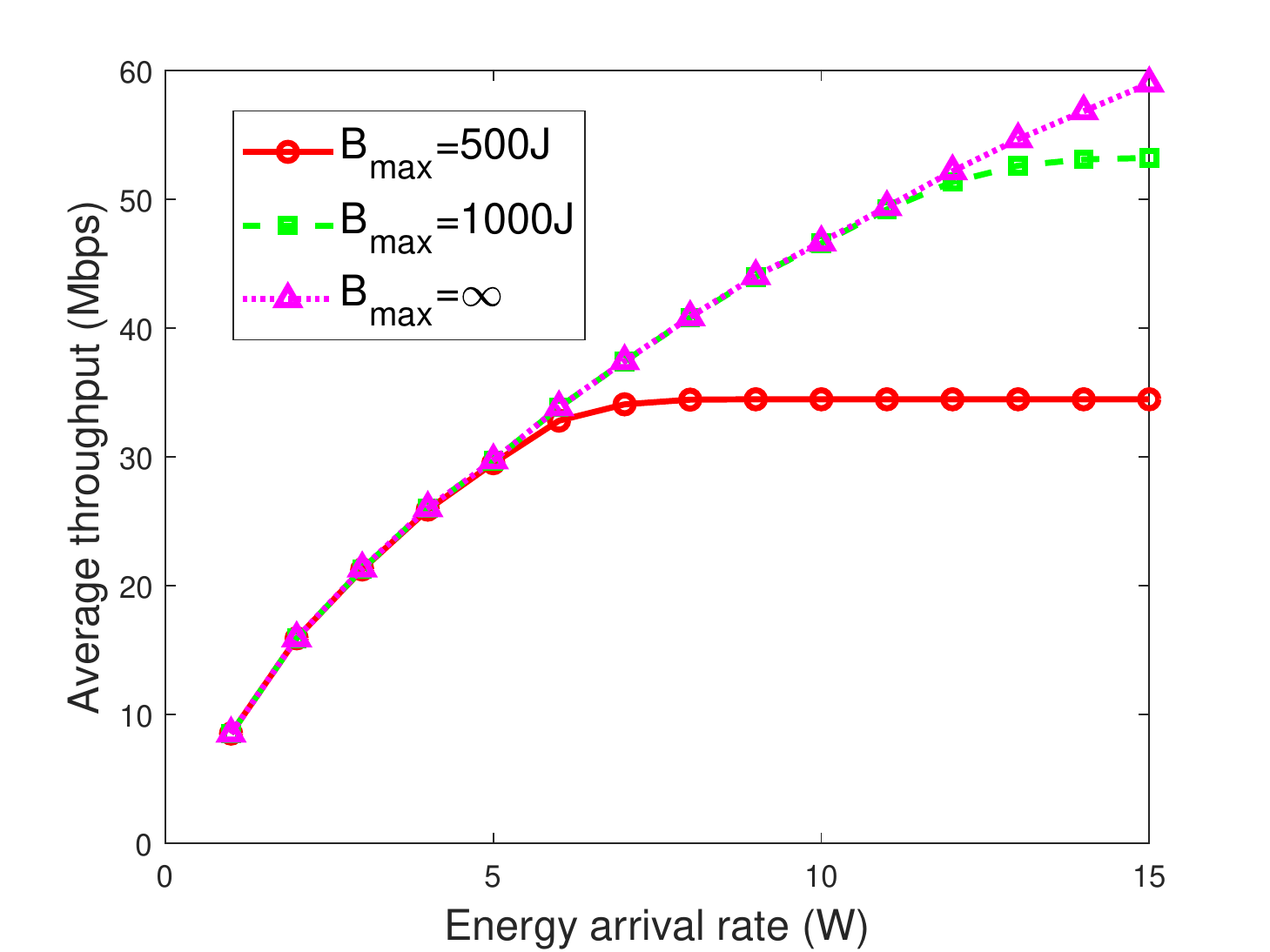}
	   \caption{{Throughput versus the energy arrival rate with flexible functional split under different battery size.}}
	   \label{fig:energy}
	\end{minipage}
\end{figure}

The relationship between the average throughput and the energy arrival rate with flexible functional split is given under different battery sizes in Fig. \ref{fig:energy}. We can see that when the energy arrival rate is small, the throughput with different battery sizes are almost the same. When the energy arrival rate is large, a small battery size leads to a larger probability of energy overflow, and less energy can transfer among different epochs, which results in a smaller throughput.

\section{Conclusions}
\label{sec:conclusion}
In this paper, we have studied the selection of the optimal functional split modes, and the corresponding transmission duration and transmission power with each mode, to maximize the throughput given the average fronthaul rate in C-RAN with renewable energy powered RRU. The optimal offline policy has the property that at most two modes should be selected in each epoch, and the sum of the transmission power and the reciprocal of the channel gain are the same for the selected functional split modes. Numerical results show that with flexible functional split, the throughput can be notably improved compared with any mode with fixed functional split. To deal with the curse of dimensionality of the online MDP problem, We derive the closed-form expression of the optimal power control policy in the scenario with one instance of energy arrival and two candidate functional split modes. We then propose a heuristic online algorithm, and numerical results show that the proposed heuristic  online policy  has similar performance with the optimal online policy.
{In the future, the optimal policy with multiple carriers will be explored, and we will further study the scenarios with multiple RRUs and random packet arrivals.}

\bibliographystyle{IEEEtran}
\bibliography{IEEEabrv,ref_energyharvesting}

\begin{thebibliography}{10}
\providecommand{\url}[1]{#1}
\csname url@samestyle\endcsname
\providecommand{\newblock}{\relax}
\providecommand{\bibinfo}[2]{#2}
\providecommand{\BIBentrySTDinterwordspacing}{\spaceskip=0pt\relax}
\providecommand{\BIBentryALTinterwordstretchfactor}{4}
\providecommand{\BIBentryALTinterwordspacing}{\spaceskip=\fontdimen2\font plus
\BIBentryALTinterwordstretchfactor\fontdimen3\font minus
  \fontdimen4\font\relax}
\providecommand{\BIBforeignlanguage}[2]{{%
\expandafter\ifx\csname l@#1\endcsname\relax
\typeout{** WARNING: IEEEtran.bst: No hyphenation pattern has been}%
\typeout{** loaded for the language `#1'. Using the pattern for}%
\typeout{** the default language instead.}%
\else
\language=\csname l@#1\endcsname
\fi
#2}}
\providecommand{\BIBdecl}{\relax}
\BIBdecl

\bibitem{mine}
L.~Wang and S.~Zhou, ``Flexible functional split in {C-RAN} with renewable
  energy powered remote radio units,'' in \emph{IEEE International Conference
  on Communications Workshops (ICC Workshops)}, Kansas City, MO, USA, May 2018,
  pp. 1--6.

\bibitem{CRANwhitepaper}
{China Mobile}, ``{C-RAN}: the road towards green {RAN},'' \emph{White Paper,
  Version 3.0}, Dec 2013.

\bibitem{CPRI}
{CPRI Specification V6.0}, ``{Common Public Radio Interface (CPRI)},''
  \emph{{Interface Specification}}, 2013.

\bibitem{SplittingBS}
U.~D{\"o}tsch, M.~Doll, H.-P. Mayer, F.~Schaich, J.~Segel, and P.~Sehier,
  ``Quantitative analysis of split base station processing and determination of
  advantageous architectures for {LTE},'' \emph{Bell Labs Technical Journal},
  vol.~18, no.~1, pp. 105--128, 2013.

\bibitem{impact}
C.-Y. Chang, N.~Nikaein, and T.~Spyropoulos, ``Impact of packetization and
  scheduling on {C-RAN} fronthaul performance,'' in \emph{IEEE Global
  Communications Conference (GLOBECOM)}, Washington, DC USA, Dec 2016, pp.
  1--7.

\bibitem{LTEmodel}
C.~Desset, B.~Debaillie, V.~Giannini, A.~Fehske, G.~Auer, H.~Holtkamp,
  W.~Wajda, D.~Sabella, F.~Richter, M.~J. Gonzalez, H.~Klessig, I.~Gódor,
  M.~Olsson, M.~A. Imran, A.~Ambrosy, and O.~Blume, ``Flexible power modeling
  of \protect{LTE} base stations,'' in \emph{IEEE Wireless Communications and
  Networking Conference (WCNC)}, Paris, France, Apr 2012, pp. 2858--2862.

\bibitem{centralize}
X.~Wang, L.~Wang, S.~E. Elayoubi, A.~Conte, B.~Mukherjee, and C.~Cavdar,
  ``Centralize or distribute? {A} techno-economic study to design a low-cost
  cloud radio access network,'' in \emph{IEEE International Conference on
  Communications (ICC)}, Paris, France, May 2017, pp. 1--7.

\bibitem{redesign}
J.~Liu, S.~Xu, S.~Zhou, and Z.~Niu, ``Redesigning fronthaul for next-generation
  networks: beyond baseband samples and point-to-point links,'' \emph{IEEE
  Wireless Communications}, vol.~22, no.~5, pp. 90--97, 2015.

\bibitem{multiplexing}
L.~Wang and S.~Zhou, ``On the fronthaul statistical multiplexing gain,''
  \emph{IEEE Communications Letters}, vol.~21, no.~5, pp. 1099--1102, 2017.

\bibitem{SDN}
M.~Chiosi, D.~Clarke, P.~Willis, A.~Reid, J.~Feger, M.~Bugenhagen, W.~Khan,
  M.~Fargano, C.~Cui, H.~Deng, and J.~Benitez, ``{Network functions
  virtualisation: An introduction, benefits, enablers, challenges and call for
  action},'' \emph{SDN and OpenFlow World Congress}, vol.~48, 202.

\bibitem{SDRandFS}
A.~Marotta, D.~Cassioli, K.~Kondepu, C.~Antonelli, and L.~Valcarenghi,
  ``Efficient management of flexible functional split through software defined
  {5G} converged access,'' in \emph{IEEE International Conference on
  Communications (ICC)}, Kansas City, MO, USA, May 2018, pp. 1--6.

\bibitem{Flex5G}
D.~Harutyunyan and R.~Riggio, ``{Flex5G: Flexible Functional Split in 5G
  Networks},'' \emph{IEEE Transactions on Network and Service Management},
  vol.~15, no.~3, pp. 961 -- 975, 2018.

\bibitem{Flex5GSurvey}
L.~M.~P. Larsen, A.~Checko, and H.~L. Christiansen, ``A survey of the
  functional splits proposed for {5G} mobile crosshaul networks,'' \emph{IEEE
  Communications Surveys \& Tutorials}, 2018.

\bibitem{EHNodes}
K.~Tutuncuoglu and A.~Yener, ``Optimum transmission policies for battery
  limited energy harvesting nodes,'' \emph{IEEE Transactions on Wireless
  Communications}, vol.~11, no.~3, pp. 1180--1189, 2012.

\bibitem{grid}
J.~Gong, S.~Zhou, and Z.~Niu, ``Optimal power allocation for energy harvesting
  and power grid coexisting wireless communication systems,'' \emph{IEEE
  Transactions on Communications}, vol.~61, no.~7, pp. 3040--3049, 2013.

\bibitem{greendelivery}
S.~Zhou, J.~Gong, Z.~Zhou, W.~Chen, and Z.~Niu, ``{GreenDelivery}: proactive
  content caching and push with energy-harvesting-based small cells,''
  \emph{IEEE Communications Magazine}, vol.~53, no.~4, pp. 142--149, 2015.

\bibitem{EHSurvey}
M.-L. Ku, W.~Li, Y.~Chen, and K.~J.~R. Liu, ``Advances in energy harvesting
  communications: Past, present, and future challenges,'' \emph{IEEE
  Communications Surveys \& Tutorials}, vol.~18, no.~2, pp. 1384--1412, 2016.

\bibitem{waterfilling}
O.~Ozel, K.~Tutuncuoglu, J.~Yang, S.~Ulukus, and A.~Yener, ``Transmission with
  energy harvesting nodes in fading wireless channels: Optimal policies,''
  \emph{IEEE Journal on Selected Areas in Communications}, vol.~29, no.~8, pp.
  1732--1743, 2011.

\bibitem{glue}
P.~Youssef-Massaad, L.~Zheng, and M.~Medard, ``Bursty transmission and glue
  pouring: on wireless channels with overhead costs,'' \emph{IEEE Transactions
  on Wireless Communications}, vol.~7, no.~12, pp. 5188--5194, 2008.

\bibitem{procost}
O.~Orhan, D.~Gunduz, and E.~Erkip, ``Energy harvesting broadband communication
  systems with processing energy cost,'' \emph{IEEE Transactions on Wireless
  Communications}, vol.~13, no.~11, pp. 6095--6107, 2014.

\bibitem{split1}
D.~A. Temesgene, N.~Piovesan, M.~Miozzo, and P.~Dini, ``Optimal placement of
  baseband functions for energy harvesting virtual small cells,'' in \emph{IEEE
  88th Vehicular Technology Conference (VTC-Fall)}, Chicago, USA, Aug 2018, pp.
  1--6.

\bibitem{split2}
D.~A. Temesgene, M.~Miozzo, and P.~Dini, ``Dynamic functional split selection
  in energy harvesting virtual small cells using temporal difference
  learning,'' in \emph{IEEE 29th Annual International Symposium on Personal,
  Indoor and Mobile Radio Communications (PIMRC)}, Bologna, Italy, Sep 2018,
  pp. 1813--–1819.

\bibitem{sleepfs}
H.~Ko and S.~Pack, ``Energy-efficient mode switching mechanism with flexible
  functional splitting in energy harvesting cloud radio access networks,''
  \emph{IEEE Access}, vol.~6, pp. 65\,078 -- 65\,087, 2018.

\bibitem{wirelessBH}
X.~Ge, H.~Cheng, M.~Guizani, and T.~Han, ``{5G} wireless backhaul networks:
  challenges and research advances,'' \emph{IEEE Network}, vol.~28, no.~6, pp.
  6--11, 2014.

\bibitem{Huang14}
C.~Huang, R.~Zhang, and S.~Cui, ``Optimal power allocation for outage
  probability minimization in fading channels with energy harvesting
  constraints,'' \emph{IEEE Transactions on Wireless Communications}, vol.~13,
  no.~2, pp. 1074--1087, 2014.

\bibitem{Gong18}
J.~Gong, Z.~Zhou, and S.~Zhou, ``On the time scales of energy arrival and
  channel fading in energy harvesting communications,'' \emph{IEEE Transactions
  on Green Communications and Networking}, vol.~2, no.~2, pp. 482--492, 2018.

\bibitem{iter}
D.~V. Djonin and V.~Krishnamurthy, ``{${Q} $-Learning Algorithms for
  Constrained Markov Decision Processes With Randomized Monotone Policies:
  Application to MIMO Transmission Control},'' \emph{IEEE Transactions on
  Signal Processing}, vol.~55, no.~5, pp. 2170 -- 2181, 2017.

\bibitem{DP}
D.~P. Bertsekas, ``Dynamic programming and optimal control,'' \emph{Athena
  Scientific Belmont}, vol.~1, no.~3, 2005.

\bibitem{eCPRI}
{eCPRI Specification V1.1 }, ``{Common Public Radio Interface: eCPRI Interface
  Specification},'' \emph{{Interface Specification}}, 2018.

\bibitem{energymodel}
G.~Lee, W.~Saad, M.~Bennis, A.~Mehbodniya, and F.~Adachi, ``Online ski rental
  for {ON/OFF} scheduling of energy harvesting base stations,'' \emph{IEEE
  Transactions on Wireless Communications}, vol.~16, no.~5, pp. 2976 -- 2990,
  2017.

\end{thebibliography}

\end{document}